\documentclass[journal,12pt,onecolumn,draftclsnofoot,]{IEEEtran}

\usepackage{times}
\usepackage{lipsum}
\usepackage{mwe}
\usepackage{fancybox}
\usepackage{multicol}
\usepackage{color}
\usepackage{graphicx}
\usepackage{epsfig}
\usepackage{graphicx,latexsym}
\usepackage{amssymb,amsthm,amsmath}
\usepackage{array}
\usepackage{xy}
\usepackage{amsthm}
\usepackage{mathtools}
\usepackage{epstopdf}
\usepackage{algorithm}
\usepackage{algpseudocode}
\usepackage{cuted}
\usepackage{booktabs}
\usepackage{cite}
\usepackage{ amssymb }
\usepackage[table,xcdraw]{xcolor}
\usepackage{float}
\usepackage{multirow}

\usepackage{url}

\newcommand{\re}{\textcolor{black}}

\renewcommand{\theenumi}{\arabic{enumi}}

\newcolumntype{C}[1]{>{\centering\arraybackslash$}p{#1}<{$}}
\makeatletter

\newcommand{\Rmnum}[1]{\expandafter\@slowromancap\romannumeral #1@}

\newtheorem{theorem}{Theorem}

\newtheorem{lemma}[theorem]{Lemma}
\newcommand{\norm}[1]{\lVert#1\rVert_2}
\newcommand{\normf}[1]{\lVert#1\rVert_F}

\makeatletter
\def\endthebibliography{%
  \def\@noitemerr{\@latex@warning{Empty `thebibliography' environment}}%
  \endlist
}
\newcommand{\algmargin}{\the\ALG@thistlm}   
\makeatother
\algnewcommand{\parState}[1]{\State%
    \parbox[t]{\dimexpr\linewidth-\algmargin}{\strut #1\strut}}
\begin{document}
\title{Energy-Efficient UAV-Assisted Mobile Edge Computing: Resource Allocation and Trajectory Optimization}
\author{Mushu~Li,~\IEEEmembership{Student~Member,~IEEE,}
		Nan~Cheng\footnotemark{},~\IEEEmembership{Member,~IEEE,}
        Jie~Gao,~\IEEEmembership{Member,~IEEE,}
        Yinlu~Wang,
        Lian~Zhao,~\IEEEmembership{Senior~Member,~IEEE,}
        and Xuemin~(Sherman)~Shen,~\IEEEmembership{Fellow,~IEEE}}
\maketitle

\footnotetext{
M. Li, J. Gao, and X. Shen are with the Department of Electrical and Computer
Engineering, University of Waterloo, Waterloo, ON N2L 3G1, Canada (e-mail:,
m475li@uwaterloo.ca; jie.gao@uwaterloo.ca; sshen@uwaterloo.ca).

N. Cheng (corresponding author) is with State Key Lab. of ISN, and with the School of Telecommunications Engineering, Xidian University, Xian 710071, China (e-mail: nancheng@xidian.edu.cn).

Y. Wang is with the National Mobile Communications Research Laboratory, Southeast University, Nanjing, China (e-mail: yinluwang@seu.edu.cn).

L. Zhao is with the Department of Electrical, Computer and Biomedical Engineering, Ryerson University, Toronto, ON M5B 2K3, Canada (e-mail: l5zhao@ryerson.ca).
} 
\begin{abstract}
In this paper, we study unmanned aerial vehicle (UAV) assisted mobile edge computing (MEC) with the objective to optimize computation offloading with minimum UAV energy consumption.  In the considered scenario, a UAV plays the role of an aerial cloudlet to collect and process the computation tasks offloaded by ground users. Given the service requirements of users, we aim to maximize UAV energy efficiency by jointly optimizing the UAV trajectory, the user transmit power, and computation load allocation. The resulting optimization problem corresponds to nonconvex fractional programming, and the Dinkelbach algorithm and the successive convex approximation (SCA) technique are adopted to solve it. Furthermore, we decompose the problem into multiple subproblems for distributed and parallel problem solving. To cope with the case when the knowledge of user mobility is limited, we adopt a spatial distribution estimation technique to predict the location of ground users so that the proposed approach can still be applied. Simulation results demonstrate the effectiveness of the proposed approach for maximizing the energy efficiency of UAV.

\end{abstract}
\section{Introduction}
  
Driven by the visions of Internet of Things (IoT) and 5G communications, mobile edge computing (MEC) is considered as an emerging paradigm that leverages the computing resource and storage space deployed at network edges to perform latency-critical and computation-intensive tasks for mobile users \cite{survey1}.  
The computation tasks generated by mobile users 
can be offloaded to the nearby edge server, such as macro/small cell base station and Wi-Fi access point, to reduce computation delay and computing energy cost at mobile devices. Moreover, by pushing the traffic, computation, and network functions to the network edges, mobile users can enjoy low task offloading time with less backhaul usage \cite{Gao2}.

{Specifically, in IoT era, MEC is considered as a key enabling technology to support the computing services for billions of IoT nodes to be deployed \cite{Wu_n,Fu2}.} Since the most of IoT nodes are power-constrained and have limited computing compatibility, they can offload their computation tasks to network edges to extend their battery life and improve the computing efficiency. However, many IoT nodes are operating in unattended or challenging areas, such as forests, deserts, mountains, or underwater locations \cite{Mohamed}, to execute some computation-intensive applications, including long pipeline infrastructures monitoring and control \cite{Khan}, underwater infrastructures monitoring \cite{Domingo}, and military operations \cite{Samad}. In these scenarios, the terrestrial communication infrastructures are distributed sparsely and cannot provide reliable communications for the nodes. Therefore, in this paper, we utilize unmanned aerial vehicles (UAVs) to provide ubiquitous communication and computing supports for IoT nodes. Equipped with computing resources, UAV-mounted cloudlet can collect and process the computation tasks of ground IoT nodes that cannot connect to the terrestrial edges. As UAVs are fully controllable and operate at a high altitude, they can be dispatched to the designated places {for providing efficient on-demand communication and computing services to IoT nodes in a rapid and flexible manner \cite{Fu, Zhou,Shi,Cheng_3}.}

Despite the advantages of UAV-assisted MEC, there are several challenges in network deployment and operation. Firstly, 
the onboard energy of a UAV is usually limited. To improve the user experience on the computing service, UAVs should maximize their energy efficiency by optimizing their computing ability in the limited service time. 
Secondly, planning an energy-aware UAV trajectory is another challenge in UAV-assisted networks. The UAV is required to move to collect the offloaded data from sparsely distributed users for the best channel quality, while a significant portion of UAV energy consumption stems from mechanical actions during flying. 
Thirdly, the computation load allocation cannot be neglected even though the computing energy consumption in UAV-mounted cloudlet is relatively small compared to its mechanical energy. In the state-of-art MEC server architecture, the dynamic frequency and voltage scaling (DVFS) technique is adopted. The computing energy for a unit time is growing cubically as the allocated computation load increases \cite{survey1}. Without proper allocation, the computing energy consumption could blow up,
or the offloaded tasks cannot be finished in time.
More importantly, 
UAV trajectory design, computation load allocation, and communication resource management are coupled in the MEC system \cite{Hu}, which makes the system even more complex. 
\re{To the best of our knowledge, the joint optimization of UAV trajectory, computation load allocation, and communication resource management considering energy efficiency has not been investigated in the UAV-assisted MEC system.}

To address the above challenges, we consider an energy constrained UAV-assisted MEC system in this paper. {IoT nodes as ground users can access and partially offload their computation tasks to the UAV-mounted cloudlet according to their service requirements.} The UAV flies according to a designed trajectory to collect the offloading data, process computation tasks, and send computing results back to the nodes. For each data collection and task execution cycle, we optimize the energy efficiency of the UAV, which is defined as the ratio of the overall offloaded computing data to UAV energy consumption in the cycle, by jointly optimizing the UAV trajectory and resource allocation in communication and computing aspects. The main contributions of the paper are summarized as follows.
\begin{enumerate}
\item {We develop a model for energy-efficient UAV trajectory design and resource allocation in the MEC system. {The model incorporates computing service improvement and energy consumption minimization in a UAV-mounted cloudlet.} The communication and computing resources are allocated subject to the user communication energy budget, computation capability, and the mechanical operation constraints of the UAV.}
\item 
{We exploit the successive convex approximation (SCA) technique and Dinkelbach algorithm to transform the non-convex fractional programming problem into a solvable form. In order to improve scalability, we further decompose the optimization problem by the alternating direction method of multipliers (ADMM) technique. UAV and ground users solve the optimization problem cooperatively in a distributed manner. 
By our approach, both users and UAV can obtain the optimal resource allocation results iteratively without sharing local information.}

\item {We further consider the scenario with limited knowledge of node mobility. A spatial distribution estimation technique, Gaussian kernel density estimation, is applied to predict the location of ground users. Based on the predicted location information, our proposed strategy can determine an energy-efficient UAV trajectory when the user mobility and offloading requests are ambiguous at the beginning of each optimization cycle.}

\end{enumerate}

The remainder of the paper is organized as follows. Related works are discussed in Section \ref{sec2}. The system model is provided in Section \ref{sec3}. Problem formulation and the corresponding approach are presented in Section \ref{sec4} and \ref{sec5}, respectively. The extended implementation of the proposed approach are provided in Section \ref{sec6}.
Finally, extensive simulation results and conclusions are provided in Sections \ref{sec7} and \ref{sec8}, respectively.

\section{Related Works}
\label{sec2}
\subsection{Mobile Edge Computing}
{To improve the user experience on mobile computing in 5G era, the concept of MEC has been proposed in \cite{WP} to reduce the transmitting and computing latency by utilizing a vast amount of computation resource located at edge devices.} 
The works \cite{Zhang, Mao} consider energy-efficient computing in MEC. 
In \cite{Zhang}, Zhang \textit{et al.} study the total energy consumption minimization in 5G heterogeneous networks. The mobile users make binary offloading decisions to determine where their computation tasks are executed. 
In \cite{Mao}, Mao \textit{et al.} investigate the MEC system with energy harvesting device and propose an online Lyapunov-based method to reduce the computing latency and the probability of task dropping.
The works \cite{Kuang,Rodrigues1,Rodrigues2} study radio resource allocation for computation offloading in edge computing. 
In \cite{Kuang}, Kuang \textit{et al.} propose a partial offloading
scheduling and power allocation approach for single user MEC system and jointly minimize the task execution delay and energy consumption in MEC server while guaranteeing the transmit power constraint of the user.
{In \cite{Rodrigues1,Rodrigues2}, Rodrigues \textit{et al.} investigate transmit power control and service migration policy to balance the computation load among edge servers and reduce the overall computing delay accordingly.}
{The above works consider resource allocation in MEC with fixed edge infrastructures. 
To provide on-demand service for remote IoTs, our work studies edge computing supported by UAV-mounted cloudlet, which introduces dynamic channel conditions and mechanical operation constraints.}

\subsection{UAV-assisted Network}
The UAV-assisted communication network has been investigated in works \cite{Wu_2,Zeng,Tang2}. 
In \cite{Wu_2}, Wu \textit{et al.} consider trajectory design and communication power control for a multi-UAV multi-user system, in which the objective is to maximize the throughput over ground users in a downlink scenario.
In \cite{Zeng}, Zeng \textit{et al.} analyze the energy efficiency of the UAV-assisted communication network and design a UAV trajectory strategy for hovering above a single ground communication terminal. {In \cite{Tang2}, Tang \textit{et al.} investigate a game-based channel assignment scheme for UAVs in D2D-enabled communication networks.}
{UAVs have also been utilized to enhance the flexibility of a MEC system in \cite{Garg,Messous}, where UAVs behave as communication relays to participate in the computation offloading process.}
Moreover, recently, more works utilize UAV as an aerial cloudlet to provide edge computing service \cite{Jeong,Tang1, Cheng}.
In \cite{Jeong},  Jeong \textit{et al.} study UAV path planning to minimize {communication energy consumption for task offloading at mobile users, where the energy consumption of UAV-mounted cloudlet is constrained}. Both orthogonal and non-orthogonal channel models are considered in the work.
{In \cite{Tang1}, Tang \textit{et al.} propose a UAV-assisted recommendation system in location based social networks (LBSNs), while a UAV-mounted cloudlet is deployed to reduce computing and traffic load of the cloud server.}
In \cite{Cheng}, Cheng \textit{et al.} provide the computation load offloading strategy in an IoT network given the pre-determined UAV trajectories. {The work aims to minimize the computing delay, user energy consumption, and server computing cost jointly, where the energy consumption of the UAV-mounted cloudlet has not been investigated.}
None of the above works discusses the energy efficiency on mobile computing in a UAV-mounted cloudlet, which is considered as a meaningful metric for prolonging the computing service lifetime.
Note that although \cite{Zeng} also studies energy-efficient trajectory design, it focuses on a single-ground-terminal scenario, whereas our work focuses on a multi-user scenario with corresponding resource management.
\section{System Model}
\label{sec3}
\begin{table}[]
\centering{
\caption{{List of Symbols}}
\begin{tabular}{@{}ll@{}}
\toprule
Symbol                               & Definition                                                                                                                     \\ \midrule
$k$                                  & index of time slot                                                                                                             \\
$i$                                  & index of ground node/user                                                                                                      \\
$\mathcal{I}$                        & set of users, where $\mathcal{I}=\{1,\dots,N\}$                                                                                                                   \\
$\mathcal{K}$                        & set of time slots, where $\mathcal{K}=\{1,\dots,K\}$                                                                                                              \\
${\mathbf{a}_k} (\mathbf{Q})$        & average acceleration of the UAV in slot $k$                                                                                        \\
$a_{max}$                            & maximum acceleration of the UAV                                                                                                    \\
$B$                                  & channel bandwidth                                                                                                              \\
${E}_{i}^{T}$                        & maximum offloading communication energy of user $i$                                                                            \\
$E_{i,k}^{C,U}(\mathbf{W}_k)$        & \begin{tabular}[c]{@{}l@{}}UAV computing energy for executing tasks from user $i$ \\ in time slot $k$\end{tabular}             \\
$E^F_{k}(\mathbf{Q})$                & UAV propulsion
energy consumption in slot $k$                                                                                         \\
$\hat{E}_{i}^{M}$                    & maximum computing energy consumption of user $i$                                                                               \\
$f^{U}_{k}(\mathbf{W}_k)$            & CPU-cycle frequency in time slot $k$                                                                                           \\
$f_i^M$                              & CPU-cycle frequency of user $i$                                                                                          \\
$h_{i,k} (\mathbf{Q}_k)$             & channel gain for user $i$ in slot $k$                                                                                          \\
$h_x$, $h_y$                         & bandwidth of the 2-D Gaussian kernel                                                                                           \\
$H$                                  & UAV flying altitude                                                                                                               \\
$I_i$                                & overall input data size for computation tasks of user $i$                                                                        \\
$\check{I}_i$                        & minimum input data amount to be offloaded for user $i$                                                                         \\
$K$                                  & number of time slots in a time window                                                                                                                \\
$N$                                  & number of users                                                                                                                \\
$P$                                  & maximum transmit power of a user                                                                                               \\
$\mathbf{q}_{i,k}$                   & horizontal coordinate of user $i$ in slot $k$                                                                                  \\
$\mathbf{Q}_k$                       & horizontal coordinate of the UAV in slot $k$                                                                                       \\
$R_{i,k}(\delta_{i,k},\mathbf{Q}_k)$ & data rate for user $i$ in slot $k$                                                                                             \\
$S_{i,k}({\delta}_{i,k})$            & communication energy for user $i$ in slot $k$                                                                                  \\
$T$                                  & time length of a computing cycle                                                                                               \\
${\mathbf{v}_k}(\mathbf{Q})$         & average velocity of the UAV in slot $k$                                                                                            \\
$v_{max}$                            & maximum velocity of the UAV                                                                                                        \\
$W_{i,k}$                            & \begin{tabular}[c]{@{}l@{}}amount of data offloaded by $i$ to be processed in \\ slot $k$ at the UAV-mounted cloudlet\end{tabular} \\
$\gamma_1$, $\gamma_2$               & UAV propulsion
energy consumption parameters                                                                                          \\
$\delta_{i,k}$                       & \begin{tabular}[c]{@{}l@{}}portion of the maximum power allocated to user $i$ \\ within slot $k$\end{tabular}                  \\
$\Delta$                             & time length of a time slot                                                                                                     \\
$\sigma^2$                           & power spectral density of channel noise                                                                                        \\
$\chi_i$                             & number of computation cycles for executing 1 bit                                                                               \\ \bottomrule
\end{tabular}}
\end{table}
\subsection{Network Model}
\begin{figure}[t]  
 \centering   
  \includegraphics[width=70mm]{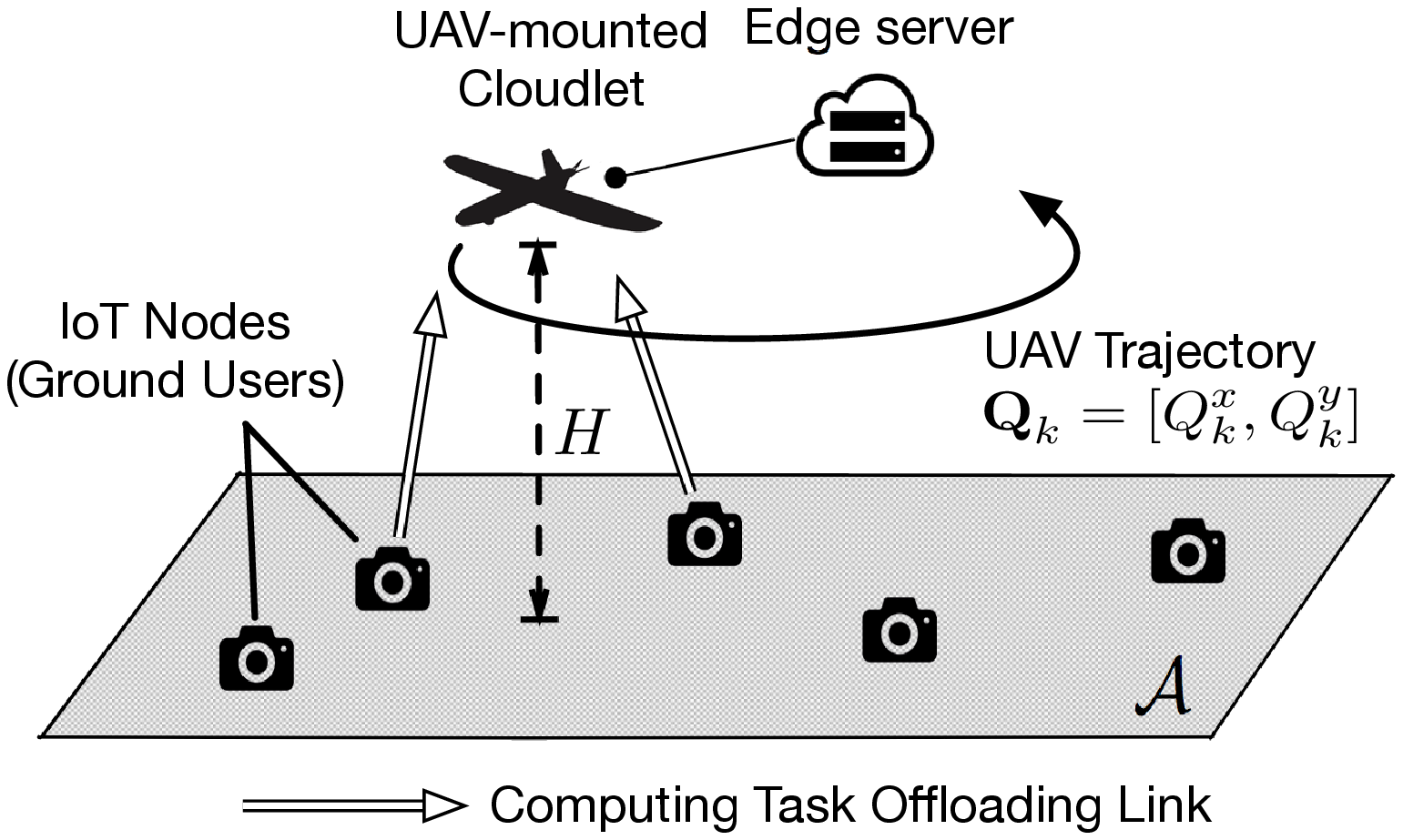}\\  
  \caption{System model.}  
\end{figure} 
{The UAV-assisted MEC system is shown in Fig. 1, in which a single UAV-mounted cloudlet is deployed to offer edge computing service for ground users in area $\mathcal{A}$. 
The UAV periodically collects and processes the computation tasks offloaded from ground users. Each user processes the rest of the computation tasks locally if the task cannot be fully collected by the UAV.} Define the computing cycle as a duration of $T$ seconds. Each cycle contains $K$ discrete time slots with equal length. Denote the set of time slots in the cycle by $\mathcal{K}$. Thus, the time length for a slot is $T/K$, which is denoted by $\Delta$. {The list of symbols is given in Table I.}

At the beginning of each cycle, ground users with computation tasks in area $\mathcal{A}$ send offloading requests to the UAV-mounted cloudlet. 
Denote the set of those ground users by $\mathcal{I}$, where $\mathcal{I}=\{1,\dots,N\}$. Assume the ground users in $\mathcal{I}$ can connect to the UAV for all time slots in the cycle. 
{In this work,} the UAV and the users cooperatively determine the offloading and resource allocation strategy for this cycle, including the UAV moving trajectory, the transmit power of ground users, and computation load allocation for UAV-mounted cloudlet. {Assume that the computation loads on solving the optimization problem are negligible compared to the computation loads of the offloaded tasks.} During the cycle, UAV flies over the ground users and offers the computing service according to the designed trajectory and resource allocation strategy. By the end of the cycle, UAV returns to a predetermined final position. 

\subsection{Communication Model}
The quality of communication links between the UAV and ground users is dependent on their location. To represent their locations, we construct a 3D Cartesian coordinate system. For IoT node $i$, the horizontal coordinate at time $k$ is denoted by $\mathbf{q}_{i,k} = [{q}_{i,k}^x, {q}_{i,k}^y]$. Assume that nodes know their trajectory for the upcoming cycle, \textit{i.e.}, $\{\mathbf{q}_{i,k}, \forall k\}$. For the UAV, the horizontal coordinate at time $k$ is denoted by $\mathbf{Q}_k = [{Q}_{k}^x,{Q}_{k}^y]$. The UAV moves at a fixed altitude $H$. The UAV trajectory plan, as an optimization variable, consists of UAV positions in the whole cycle, \textit{i.e.}, $\mathbf{Q} = [\mathbf{Q}_1; \dots; \mathbf{Q}_K]$. The average UAV velocity in slot $k$ is given by
\begin{equation}
{\mathbf{v}_k}(\mathbf{Q}) = \frac{{\mathbf{Q}_{k}-\mathbf{Q}_{k-1}}}{\Delta}, \forall k.
\end{equation}
The average acceleration in slot $k$ is given by
\begin{equation}
{\mathbf{a}_k} (\mathbf{Q}) = \frac{{\mathbf{v}_{k} (\mathbf{Q})-\mathbf{v}_{k-1} (\mathbf{Q})}}{\Delta}, \forall k.
\end{equation}
{The magnitudes of velocity and acceleration are constrained by the maximum speed and acceleration magnitude, which are denoted by $v_{\textrm{max}}$ and $a_{\textrm{max}}$, respectively.}

It is assumed that the doppler frequency shift in the communication can be compensated at the receiver. The channel quality depends on the distance between the UAV and users. Due to the high probability of LOS links in UAV communication\cite{Zeng}, we assume that the channel gain follows a free-space path loss model. The channel gain for user $i$ in slot $k$ is denoted by $h_{i,k}$, where 
\begin{equation}
h_{i,k} (\mathbf{Q}_k) = \frac{g_0}{\norm{\mathbf{Q}_k-\mathbf{q}_{i,k} }^2 + H^2},
\end{equation}
{where $\norm{\cdot}$ is the notation representing the L2 norm.}
The parameter $g_0$ denotes the received power at the reference distance (e.g., $d = 1$ m) between the transmitter and the receiver. 
We consider two channel access schemes: i) orthogonal access, in which the bandwidth is divided into $N$ sub-channels each occupied by one user; and ii) non-orthogonal access, in which the frequency bandwidth is shared among users. Denote the channel bandwidth for the uplink by $B$. The amount of data that can be offloaded by user $i$ in slot $k$ is
\begin{equation}
R_{i,k}(\delta_{i,k},\mathbf{Q}_k) = \frac{B \Delta}{N} \log\big[1+\frac{\delta_{i,k}h_{i,k}(\mathbf{Q}_k)P}{\sigma^2(B/N)}\big],
\end{equation}
under the orthogonal access model, and,
\begin{equation}
R_{i,k}(\boldsymbol{\delta}_{k}, \!\mathbf{Q}_k) \! = \! {B \Delta}  \!\log\big[1 + \!\frac{\delta_{i,k}h_{i,k}(\mathbf{Q}_k)P}{\sigma^2B \!+ \!\sum_{j \neq i} \!\delta_{j,k}h_{j,k}(\mathbf{Q}_k)P}\big],
\end{equation} 
under the non-orthogonal channel model. The parameter $P$ and $\sigma^2$ denote the maximum transmit power of ground users and the power spectral density of channel noise, respectively. The variable $\delta_{i,k} \in [0,1]$ represents the portion of the maximum power that is allocated to user $i$ within time slot $k$, which is a part of the offloading strategy. The symbol $\boldsymbol{\delta}_{k}$ denotes the vector of $\delta_{i,k}$ for all $i\in\mathcal{I}$ in slot $k$. The noise power in the transmission is represented by $n_0$, where $n_0 = \sigma^2B/N$ for the orthogonal channel access model, and $n_0 = \sigma^2B$ for the non-orthogonal channel access model. {In non-orthogonal model, users share the same channel to offload their tasks. The communication power allocated for a user will interfere the data rate of other users. }

\subsection{Computation Model}
Due to the limited battery and the computing capability of the UAV, only a part of tasks can be offloaded and executed in the UAV-mounted cloudlet. {Full granularity in task partition is considered, where the task-input data can be arbitrarily divided for local and remote executions \cite{Jeong, Wang2,Wang}.} Accordingly, a portion of the computation tasks are offloaded to the cloudlet while the rest are executed by the ground users locally. Users upload the input data for their offloaded tasks, and the UAV processes the corresponding computation loads of those tasks. 
Assume that the computation load can be executed once the input data is received, and the computing data amount is equal to the input data amount of tasks \cite{Jeong}. A task partition technique is considered, where the partition of the computation input bits are utilized to measure the division between the offloaded computation load and local computation load.   
The overall input data size for computation tasks of user $i$ is denoted by $I_i$. We set the threshold $\check{I}_i$ as the minimum input data amount required to be offloaded to the cloudlet for user $i$, where $\check{I}_i\leq I_i$. The threshold represents the part of computation tasks having to be conducted in the cloudlet. Thus, the overall offloaded bits of user $i$ is constrained as follows:
\begin{equation}
\check{I}_i\leq \sum_{k\in \mathcal{K}} R_{i,k}(\boldsymbol{\delta}_{k}, \mathbf{Q}_k) \leq I_i, \forall i. \label{eq.c5}
\end{equation}
Under the scenario that the threshold is satisfied, if user's tasks cannot be fully offloaded, the rest of the tasks are processed by IoT nodes locally.

After users upload the input data, the UAV will save the received data to a buffer with enough capacity for further processing. 
The UAV processes the received data according to the workload allocation results. {Let the variable $W_{i,k}$ denote the amount of data, which is from user $i$'s offloaded task, to be processed in slot $k$.} The UAV can only compute the task which is offloaded and received, and all offloaded tasks should be executed by the end of the cycle. Therefore, the following computation constraints are given:
\begin{subequations}
\begin{equation}
\sum_{t = 1}^k R_{i,t}(\boldsymbol{\delta}_{k}, \mathbf{Q}_k)  \geq \sum_{t = 1}^k W_{i,t}, \forall k \label{eq.c1}
\end{equation}    
\begin{equation}
\sum_{t = 1}^K R_{i,t}(\boldsymbol{\delta}_{k}, \mathbf{Q}_k)  = \sum_{t = 1}^K W_{i,t}.\label{eq.c2}
\end{equation}
\end{subequations}

In addition, for the local computing, the CPU-cycle frequency of the IoT node $i$ is fixed as $f_i^M$. For the UAV-mounted cloudlet, 
we consider the CPU featured by DVFS technique. The CPU-cycle frequency can step-up or step-down according to the computation workload and is bounded by the maximum CPU-cycle frequency $f^{U}_{max}$. As given in \cite{survey1, Wang2}, the CPU-cycle frequency for the cloudlet can be calculated by
\begin{equation}
f^{U}_{k} (\mathbf{W}_k) = \dfrac{\sum_i \chi_i W_{i,k}}{\Delta}\leq f^{U}_{max}, \forall k,\label{eq.c6}
\end{equation}
where $f^{U}_{k}(\mathbf{W}_k)$ represents the CPU-cycle frequency in time slot $k$, and $\chi_i$ denotes the number of computation cycles needed to execute 1 bit of data.

\subsection{Energy Consumption Model}
\subsubsection{Energy Consumption at Nodes}
The main energy consumption of nodes are the energy cost from communication and local computing. Firstly, the communication energy for user $i$ offloading tasks in slot $k$ can be formulated as
\begin{equation}
S_{i,k}({\delta}_{i,k}) = \delta_{i,k}P\Delta .
\end{equation}
The overall offloading communication energy of user $i$ is bounded by  ${E}_{i}^{T}$, \textit{i.e.}, 
\begin{equation}
\sum_{k} S_{i,k}({\delta}_{i,k})\leq {E}_{i}^{T} , \forall i. \label{eq.c4}
\end{equation}
Therefore, the energy consumption of a user on communication can be reduced if the UAV is closer. On the other hand, for the computing energy consumption, we consider that the lower bound of offloaded bits $\check{I}_i$ guarantees the local computing energy under the user's computing energy requirement, \textit{i.e.}, 
\begin{equation}
{E}_{i}^{M} =  \kappa\chi_i (I_i-\check{I}_i) (f_i^M)^2 \leq \hat{E}_{i}^{M},
\end{equation}
where ${E}_{i}^{M}$ is the maximum computing energy that could be reached by threshold $\check{I}_i$, and $\hat{E}_{i}^{M}$ is the parameter representing the constraint of the computing energy consumption. The computing energy model is adopted from \cite{survey1, Yuan}. Parameters $f_i^M$ and $\kappa$ represent the fixed  CPU-cycle frequency of user $i$ and a constant related to the hardware architecture, respectively.

\subsubsection{Energy Consumption at UAV-mounted Cloudlet}
The main energy consumption at the UAV-mounted cloudlet consists of the energy cost from mechanical operation and computing. Although downlink transmission exists in our system, this part of energy consumption is negligible for two reasons: 1) The communication energy is too small compared to the UAV propulsion and computing energy. 2) The output computing results usually have much less data amount compared to the input data amount \cite{Li}. We adopt the refined UAV propulsion energy consumption model for fixed-wing UAV following \cite{Zeng} \footnote{We deploy the fixed-wing UAV in the proposed system as an example. The proposed approach also can be adapted to the system with a quad-rotor UAV, where only the mechanical
energy consumption model is different.}. The propulsion
energy consumption in slot $k$ relates to the instantaneous UAV acceleration and velocity, which is given by
\begin{equation}
E^F_{k}(\mathbf{Q}) = \gamma_1 \norm{{\mathbf{v}_k}(\mathbf{Q})}^3 + \frac{\gamma_2}{\norm{{\mathbf{v}_k}(\mathbf{Q})}}(1+\frac{\norm{a_k(\mathbf{Q})}^2}{g^2}),
\end{equation}
where $g$ denotes the gravitational acceleration. {$\gamma_1$ and $\gamma_2$ are fixed parameters related to the aircraft's weight, wing area, air density, etc.} The value of parameters is given in \cite{Zeng,Jeong}. 
The computing energy for executing tasks from user $i$ in time slot $k$ is expressed as
\begin{equation}
E_{i,k}^{C,U}(\mathbf{W}_k) = \kappa \chi_i W_{i,k}\left(f_k^U(\mathbf{W}_k))^2. \right. \label{eq.c60}
\end{equation} 

\section{Problem Formulation}
\label{sec4}

In this work, the main objective is to maximize the energy efficiency of the UAV-mounted cloudlet subject to user offloading constraints, UAV computing capabilities, and the mechanical constraints of the UAV. The energy efficiency of the UAV is defined as the ratio between the overall offloaded data and the energy consumption of the UAV in a cycle. The energy efficiency maximization problem is formulated as follows.
\begin{subequations}
\label{eq:optProb2}
\begin{align}
&\max_{\boldsymbol{\delta}, \mathbf{W}, \mathbf{Q}} && \eta = \frac{\sum_{i\in \mathcal{I}} \sum_{k\in \mathcal{K}} R_{i,k}(\boldsymbol{\delta}_{k}, \mathbf{Q}_k)}{\sum_{k\in \mathcal{K}}\sum_{i\in \mathcal{I}}E_{i,k}^{C,U}(\mathbf{W}_k) +\sum_{k\in \mathcal{K}}E_{k}^{F}(\mathbf{Q})} \tag{14}\\
&  \text{s.t.}&&  \norm{ \mathbf{v}_k(\mathbf{Q})} \leq v_{max}, \forall k, \label{eq.c7} \\ 
&&& \norm{ \mathbf{a}_k(\mathbf{Q})} \leq a_{max}, \forall k, \label{eq.c9} \\ 
&&& \mathbf{Q}_{K} = \mathbf{Q}_f, {\mathbf{v}_K(\mathbf{Q})} = {\mathbf{v}_0},\label{eq.c8}\\ 
&&& 0\leq \delta_{i,k}\leq 1,\label{eq.c10} \\ 
&&&  (\ref{eq.c5}), (\ref{eq.c1}),(\ref{eq.c2}), (\ref{eq.c6}),(\ref{eq.c4}). \notag
\end{align}
\end{subequations}
The term $\mathbf{Q}_f$ represents the designated final position of the UAV, and $\mathbf{v}_0$ represents the initial velocity at the beginning of the cycle. The constraints can be categorized into three types: 1)  user QoS constraints, including (\ref{eq.c5}), (\ref{eq.c4}), and (\ref{eq.c10}); 2) UAV computing ability constraints, including (\ref{eq.c1}), (\ref{eq.c2}), and (\ref{eq.c6}); 3) UAV mechanical constraints, including (\ref{eq.c7}), (\ref{eq.c9}), and (\ref{eq.c8}).
The optimization problem is a non-linear fractional programming. {In addition, due to the interference among users in the non-orthogonal channel and the propulsion energy consumption for the fixed-wing UAV, both functions $R_{i,k}(\boldsymbol{\delta}_{k} , \mathbf{Q}_k)$ and $E_{k}^{F}(\mathbf{Q})$ are non-convex.
Therefore, solving optimization problem (\ref{eq:optProb2}) is challenging.} To search the global optimizer of a non-convex problem is often slow and may not be feasible. In the following section, we will propose an approach to find a local optima efficiently.
\section{Proposed Optimization Approach}
\label{sec5}
In this section, an optimization approach is introduced to find a solution of problem (\ref{eq:optProb2}). Firstly, an inner convex approximation method is applied to approximate the non-convex functions $R_{i,k}(\boldsymbol{\delta}_{k} , \mathbf{Q}_k)$ and $E_{k}^{F}(\mathbf{Q})$ by solvable convex functions. The SCA-based algorithm is adopted to achieve the local optimizer of the original problem.
After the approximated convex functions are built, the fraction programming in the inner loop of the SCA-based algorithm is handled by the Dinkelbach algorithm. Moreover, in order to improve scalability, the problem is further decomposed into several sub-problems via ADMM technique, in which the power allocation is solved by users in a distributed manner, while the computation load allocation and UAV trajectory planning are determined by UAV itself. 
The details are presented in following subsections.
\subsection{Successive Convex Approximation}
Problem (\ref{eq:optProb2}) is a non-convex problem due to $R_{i,k}(\boldsymbol{\delta}_{k} , \mathbf{Q}_k)$ and $E_{k}^{F}(\mathbf{Q})$. To construct an approximation that is solvable, we first introduce several auxiliary variables, $\{\xi_{i,k},\omega_k, l_{i,k}, A_k,\check{{R}}_{i,k},\hat{{E}}_{i,k}^{F}\}$. For the orthogonal channel access scheme, the new optimization problem is shown as follows:

\begin{subequations}
\label{eq:optProb3}
\begin{align}
&\max_{\mathcal{V}} && \check\eta(\mathcal{V}) = \frac{\sum_{i\in \mathcal{I}} \sum_{k\in \mathcal{K}} \check{R}_{i,k}}{\sum_{k\in \mathcal{K}}\sum_{i\in \mathcal{I}}E_{i,k}^{C,U}(\mathbf{W}_k) +\sum_{k\in \mathcal{K}}\hat{E}_{k}^{F}}	\tag{15} \label{obj:optProb3}\\
&  \text{s.t.}&&  \check{R}_{i,k} \leq \frac{B \Delta}{N} \log(1+\xi_{i,k}), \forall i, k  \label{eq.c11}\\
&&& \xi_{i,k} l_{i,k} \leq {\delta_{i,k}P}, \forall i, k  \label{eq.c12}\\
&&& \frac{(\norm{\mathbf{Q}_k-\mathbf{q}_{i,k} }^2 + H^2)n_0}{g_0} \leq l_{i,k}, \forall i, k  \label{eq.c13}\\
&&&  \hat{E}_{k}^{F} \geq \gamma_1 \norm{{\mathbf{v}_k}(\mathbf{Q})}^3 + \gamma_2 A_k , \forall  k  \label{eq.c14}\\
&&& \omega_k^2 \leq \norm{\mathbf{v}_k(\mathbf{Q})}^2 , \forall  k \label{eq.c15}\\
&&& \omega_k A_k \geq 1+\frac{\norm{a_k(\mathbf{Q})}^2}{g^2}, \forall  k \label{eq.c16} \\
&&& \check{I}_i\leq \sum_{k\in \mathcal{K}} \check{R}_{i,k}\leq I_i, \forall i, \label{eq.c18}\\
&&&  (\ref{eq.c5}), (\ref{eq.c1}),(\ref{eq.c2}), (\ref{eq.c6}),(\ref{eq.c4}), (\ref{eq.c7})-(\ref{eq.c10}).\notag
\end{align}
\end{subequations}
Set $\mathcal{V}$ represents the union set of the primary and auxiliary optimization variables, where  $\mathcal{V} = \{ \boldsymbol{\delta}, \mathbf{W}, \mathbf{Q}, \boldsymbol{\xi},\boldsymbol{\omega},$ $\mathbf{l},\mathbf{A},\check{\mathbf{R}},\hat{\mathbf{E}}^{F}\}$. 
For the non-orthogonal channel model, constraint (\ref{eq.c11}) is replaced by the following constraint: 
\begin{equation}
\check{R}_{i,k} \leq {B \Delta} \big[\log(1+\sum_{i\in \mathcal{I}}\xi_{i,k})-  \log(1+\sum_{j\in {\mathcal{I}/\{i\}}}\xi_{j,k})\big], \forall i, k.\tag{15h} \label{eq.c11_no} 
\end{equation}

\begin{lemma}
\label{lemma.1}
Problem (\ref{eq:optProb3}) is an equivalent form of problem~(\ref{eq:optProb2}).
\end{lemma}
\begin{proof}
See Appendix A.
\end{proof}

Problem (\ref{eq:optProb3}) includes four non-convex constraints, which are (\ref{eq.c12}), (\ref{eq.c15}), (\ref{eq.c16}), and (\ref{eq.c11_no}).
We approximate those non-convex constraints by their first order Taylor expansions and adopt the successive convex optimization technique to solve the problem. 
New auxiliary variables, $\{\xi^t_{i,k}, l^t_{i,k}, \omega^t_{k},A^t_{k},\mathbf{v}_k, z^t_{i,k}\}$, are introduced to represent the corresponding estimated optimizers at the previous iteration of optimization, \textit{i.e.}, iteration $t$. The SCA-based algorithm iterates until the estimated solution reaches to a local optimizer. {Constraint (\ref{eq.c12}) can be approximated as follows:
\begin{equation}
\norm{\xi_{i,k}+l_{i,k}, {\xi^t_{i,k}-l^t_{i,k}}, x_{i,k} -1} \leq x_{i,k}+1,\label{eq.c19}
\end{equation}
where
\begin{equation}
x_{i,k} = \delta_{i,k}P -\frac{(\xi^t_{i,k}-l^t_{i,k})(\xi_{i,k}-l_{i,k})}{2}. \notag
\end{equation}
Constraint (\ref{eq.c15}) can be approximated as follows:
\begin{equation}
\omega_{k}^2 \leq \norm{\mathbf{v}_k^t}^2+2(\mathbf{v}_k^t)^T (\mathbf{v}_k(\mathbf{Q})-\mathbf{v}_k^t).\label{eq.c21}
\end{equation}
Constraint (\ref{eq.c16}) can be approximated as follows:
\begin{equation}
\norm{\omega_{k}-A_{k}, {\omega^t_{k}+A^t_{k}}, y_{k} -1, 2, \frac{2a_k(\mathbf{Q}_k)}{g}} \leq y_{k}+1, \label{eq.c20}
\end{equation}
where
\begin{equation}
y_{k} = \frac{(\omega^t_{k}+A^t_{k})(\omega_{k}+A_{k})}{2}. \notag
\end{equation}
Constraint (\ref{eq.c11_no}) can be approximated as follows:
\begin{equation}
\check{R}_{i,k}\! \leq \!\frac{B \Delta}{N} \!\big[\log(1+\xi_{i,k}+e_{i,k})- \log(1+e^t_{i,k})-\frac{e_{i,k}-e^t_{i,k}}{\ln2(1+e^t_{i,k})}\big], \label{eq.c22}
\end{equation}
where $e_{i,k} = \sum_{j\in {\mathcal{I}/\{i\}}} \xi_{i,k}$.}

\begin{lemma}
\label{lemma.2}
Non-convex constraints (\ref{eq.c12}), (\ref{eq.c15}), (\ref{eq.c16}), and (\ref{eq.c11_no}) can be approximated by the convex forms in (\ref{eq.c19})-(\ref{eq.c22}).
The solution of the approximated problem is a local maximizer of problem (\ref{eq:optProb2}), which provides the lower bound of the maximum energy efficiency that can be achieved.
\end{lemma}
\begin{proof}
See Appendix B.
\end{proof}
Based on Lemma \ref{lemma.1} and Lemma \ref{lemma.2}, the SCA-based algorithm is summarized by Algorithm \ref{al.1}. The term ${\check\eta(\mathcal{V}; \mathbb{A}^{t})}$ represents the energy efficiency $\check\eta(\mathcal{V})$ in (\ref{eq:optProb3}) with the given value in auxiliary variable set $\mathbb{A}^{t}$.                         
Note that the approximated problem inside the loop (Steps 3 and 4 in Algorithm 1) is {a fractional programming problem and still non-convex.} 
We will provide the optimal solution of the approximated problem in the remainder of the section. The convergence of SCA has been proven in \cite{Boyd}, and the algorithm will stop after finite iterations if the local optimizer exists.
\begin{algorithm}
\caption{SCA-based Algorithm for Solving Problem (\ref{eq:optProb3})} 
\label{al.1}
\begin{algorithmic}[1]
\State {Initialize the auxiliary variables $\mathbb{A}^0 = \{\xi_{i,k}^0,\omega_k^0, l_{i,k}^0, A_k^0,\check{{R}}_{i,k}^0,\hat{{E}}_{i,k}^{F,0}\}$ and loop index $t=0$.}
\Repeat
\parState {Solve the approximated problem (\ref{obj:optProb4}) for given $\mathbb{A}^t$, and denote the optimal solution for auxiliary variables by $\mathbb{A}^{t+1}$: }
{\begin{align}
&\max_{\mathcal{V}} && \check\eta(\mathcal{V}; \mathbb{A}^{t}) \label{obj:optProb4}\\
&  \text{s.t.}&&  (\ref{eq.c5}),  (\ref{eq.c1}),(\ref{eq.c2}), (\ref{eq.c6}),(\ref{eq.c4}), (\ref{eq.c7})-(\ref{eq.c10}),  \notag\\ 
&&& (\ref{eq.c13}), (\ref{eq.c14}), (\ref{eq.c18}), (\ref{eq.c19})-(\ref{eq.c21}), \notag\\
&&& (\ref{eq.c11})\textrm{, in the case of orthogonal channel}, \notag\\
&&& (\ref{eq.c22}) \textrm{, in the case of non-orthogonal channel}.\notag
\end{align}}
\State Update $t = t+1$.
\Until{ The difference of the solutions between two adjacent iterations, \textit{i.e.},  $\|\mathbb{A}^{t+1} - \mathbb{A}^{t}\|$, is below a threshold $\theta_1$.}
\end{algorithmic}
\end{algorithm}

\subsection{Dinkelbach Algorithm}
Problem (\ref{obj:optProb4})  is a fraction programming problem. We can adopt the Dinkelbach algorithm to achieve the optimal solution. The objective function (\ref{obj:optProb4}) can be rewritten as the following parametric programming form:
\begin{align}
&F^t(\alpha) = \max_{\mathcal{V}}\Big\{\sum_{k\in \mathcal{K}}\sum_{i\in \mathcal{I}}\check{R}_{i,k} - \alpha \big[\sum_{k\in \mathcal{K}}\sum_{i\in \mathcal{I}}E_{i,k}^{C,U}(\mathbf{W}_k) \notag\\
& \hspace{4.5cm}+\sum_{k\in \mathcal{K}}\hat{E}_{k}^{F} \big]|{\mathcal{V} \in \mathcal{F}^t}\Big\}, \label{eq.dickel}
\end{align}
where $\mathcal{F}^t$ represents the feasible set of problem (\ref{obj:optProb4}) at the $t$-th iteration in Algorithm 1. The function $F^t(\alpha)$ is a monotonic decreasing function of $\alpha$. Let the term $\alpha^*$ denote the solution of $F^t(\alpha^*)=0$. {Due to the monotone decreasing property of $F^t(\alpha)$}, $F^t(\alpha^*)=0$ if and only if $\alpha^*$ is equal to the optimal result of problem (\ref{obj:optProb4}), \textit{i.e.}, $\alpha^* = \check{\eta}(\mathcal{V}^*; \mathbb{A}^{t})$~\cite{Dinkelbach}. The algorithm for solving problem (\ref{obj:optProb4}) is shown in Algorithm~\ref{al.2}.
\begin{algorithm}
\caption{Dinkelbach Algorithm for Solving Problem (\ref{obj:optProb4})} 
\label{al.2}
\begin{algorithmic}[1]
\State {Initialize $\alpha^0 = 0$ if $t = 0$, $\alpha^0 = \alpha^*$ in loop $t-1$ if $t \geq 0$, and the loop index $m = 0$.}
\Repeat
\State Solve problem (\ref{eq.dickel}) for given $\alpha^m$, and denote the solution for the problem by $\mathcal{V}^{m}_d$.
\State Update the Dinkelbach auxiliary variable $\alpha^{m+1} = \check{\eta}(\mathcal{V}^{m}_d; \mathbb{A}^{t})$.
\State $m = m+1$.
\Until{$F^t (\alpha^{m+1}) \leq \theta_2$. }
\end{algorithmic}
\end{algorithm}

Due to the nature of the SCA-based algorithm and Dinkelbach algorithm, we can further cut the iteration times based on the following Lemma.
\begin{lemma}
\label{lemma4}
Denote the optimal Dinkelbach parameter $\alpha^*$ for two consecutive SCA iterations by $\alpha^*(t-1)$ and $\alpha^*(t)$. 
We have $\alpha^*(t-1) \leq \alpha^*(t)$, and $F^t(\alpha^*(t-1)) \geq F^t(\alpha^*(t)) = 0$.
\end{lemma}
\begin{proof}
Denote the optimization results and the corresponding Dinkelbach parameter at iteration $t-1$ by $\mathcal{V}^*{(t-1)}$ and $\alpha^*(t-1)$, respectively.
From Dinkelbach algorithm, we have $\alpha^*(t-1) = \check{\eta}^*(\mathcal{V}^*(t-1); \mathbb{A}^{t-1}) \leq \check{\eta}^*(\mathcal{V}^*)$. As shown in Lemma \ref{lemma.2}, the approximated function provides the global lower bound of the original optimization function, and the results have to be inside the feasible set of the approximate optimization function for the next iteration. Thus, $\check{\eta}^*(\mathcal{V}^*(t-1); \mathbb{A}^{t-1})\leq \check{\eta}(\mathcal{V}^*(t-1); \mathbb{A}^{t}) \leq \check{\eta}^*(\mathcal{V}^{*}(t); \mathbb{A}^{t})$. Therefore, $\alpha^*(t-1) \leq \alpha^*(t)$. Moreover, due to the monotonically decreasing nature of $F(\alpha)$, $F^t(\alpha^*(t-1)) \geq F^t(\alpha^*(t)) = 0$.
\end{proof}
Given Lemma \ref{lemma4}, the initial point in iteration $t$, \textit{i.e., $\alpha^0(t)$}, in Algorithm \ref{al.2} can be set at $\alpha^*(t-1)$ rather than $0$ so that the computation efficiency of the optimization algorithm can be further improved.
\subsection{Sub-problem Decomposition by ADMM}

By now, the UAV computation energy efficiency maximization problem has been transformed into a solvable form. However, solving problem (\ref{eq.dickel}) is time-consuming due to multiple second order cone (SOC) constraints and requires the local information exchange between the UAV and users. Therefore, we propose a distributed solution, in which users maximize their offloaded computation tasks in parallel while the UAV aims to minimize its energy consumption. The original problem is decomposed into several sub-problems {without losing optimality}, and the UAV and users solve the optimization problem cooperatively. 
Local information, such as the mobility of users and the propulsion energy consumption function of the UAV, is not required to be shared among users and the UAV. 

We adopt ADMM technique to decompose problem (\ref{eq.dickel}) \cite{ADMM}. The optimization solution is achieved in an iterative manner. Firstly, we introduce an auxiliary variable, $\mathbb{G}$, which is solved by users:
\[
\mathbb{G}
=
\begin{bmatrix}
    \ddot{\mathbf{Q}}_{1,1} &\dots& \ddot{\mathbf{Q}}_{N,1}& W_{1,1} &\dots& W_{N,1} \\
    \vdots & \ddots &\vdots&\vdots & \ddots &\vdots\\
    \ddot{\mathbf{Q}}_{1,K} &\dots& \ddot{\mathbf{Q}}_{N,K}& W_{1,K} &\dots& W_{N,K} 
\end{bmatrix}^T
\]
where $\ddot{\mathbf{Q}}_{i,k}$ denotes the UAV location in time slot $k$ expected by user $i$. Each user solves a part of the matrix $\mathbb{G}_i=$[$\ddot{\mathbf{Q}}_{i, 1}, W_{i, 1}; \dots; \ddot{\mathbf{Q}}_{i, K}, W_{i, K}$], and updates it to the UAV. 
Then, the UAV generates its trajectory, $\mathbf{Q}$, and overall computation load allocation according to the uploaded matrix $\mathbb{G}$. Denote the overall amount of computation load processed in slot $k$ at UAV by $V_k$, where $\mathbf{V} = [V_1; \dots; V_K]$. The results determined by the UAV are summarized by matrix $\mathbb{H}$, where $\mathbb{H} = [\mathbb{I}_{(N\times 1)}\mathbf{Q};\mathbf{V}]$. $\mathbb{I}_{(N\times 1)}$ is a vector where all $N$ entries are 1.
By the end of the ADMM algorithm, the expected UAV trajectories should be unified and follow the flying constraints. The computation load should be allocated under the UAV computing capability. Thus, in the final optimal solution, the following constraint should be satisfied:
\begin{equation}
\mathbb{P}^T \mathbb{G} = \mathbb{H}, \label{eq.admm1}
\end{equation}
where 
\[
\mathbb{P}
=
\begin{bmatrix}
    \mathbb{I}_{(N \times N)} & \textbf{0}_{(N \times 1)} \\
    \textbf{0}_{(N \times N)} & \boldsymbol{\chi}
\end{bmatrix}
\]
The vector $\boldsymbol{\chi}$ represents the computation intensity for users' tasks, where $\boldsymbol{\chi} =  [\chi_{1};\dots; \chi_{N}]$. 
The sub-matrices $\mathbb{I}_{(N \times N)}$ and $\textbf{0}_{(N \times N)}$ denote $N$-by-$N$ identity matrix and zero matrix, respectively.

In addition, for the non-orthogonal channel model, we introduce another auxiliary variable, $e_{i,k}$, which denotes the summation of $\xi_{j,k}$ in all other users except user $i$.
This variable is used to decouple the correlated $\xi_{j,k}$ in (\ref{eq.c11_no}) to facilitate the independent optimization process at each user. At the end of the optimization, $e_{i,k}$ should be equal to $\sum_{j\in {\mathcal{I}/\{i\}}}\xi_{j,k}$. For simplicity of presentation, we transform this constraint as follows:
\begin{equation}
\frac{1}{N}(e_{i,k}+\xi_{i,k})= \bar{\xi}_{k},\label{eq.admm2}
\end{equation}
where $\bar{\xi}_{k}$ is the mean of $\{\xi_{1,k},\dots, \xi_{N,k}\}$. 
{Then, the augmented Lagrangian function is formulated as follows:
\begin{align}
&{\Gamma(\mathcal{V}_\textrm{A})=-\sum_{k\in \mathcal{K}}\sum_{i\in \mathcal{I}}\check{R}_{i,k} + \alpha \big[\sum_{k\in \mathcal{K}}\sum_{i\in \mathcal{I}}E_{i,k}^{C,U}(\mathbf{W}) +\sum_{k\in \mathcal{K}}\hat{E}_{k}^{F} \big]}\notag \\
&{\hspace{2cm}+ \textrm{Tr}\big\{\mathbf{U}_1^T(\mathbb{P}^T \mathbb{G} - \mathbb{H}) \big\}  + \frac{\rho_1}{2}\normf{\mathbb{P}^T \mathbb{G} - \mathbb{H}}^2}\notag \\
&{\hspace{2cm}+\varpi\sum_{k\in \mathcal{K}}\sum_{i\in \mathcal{I}}\Big\{ U_{2,i,k}[\frac{1}{N}(e_{i,k}+\xi_{i,k})- \bar{\xi}_{k}]} \notag \\
&{\hspace{2cm}+ \frac{\rho_2}{2} [\frac{1}{N}(e_{i,k}+\xi_{i,k})- \bar{\xi}_{k}]^2\Big\}}, \label{eq.al}
\end{align}
where $\normf{\cdot}$ is the notation representing the Frobenius norm. 
Set $\mathcal{V}_\textrm{A}$ represents variables $\{\mathcal{V},\mathbb{G},\mathbb{H},{\mathbf{U}_1},{\mathbf{U}_2}\}$. } 
Variables $\mathbf{U}_1 \in \mathbf{R}^{(N+1)\times K}$ and $\mathbf{U}_2 \in \mathbf{R}^{N\times K}$ are Lagrangian multipliers for the two auxiliary constraints, (\ref{eq.admm1}) and (\ref{eq.admm2}), respectively. Two parameters, $\rho_1$ and $\rho_2$, are penalty parameters. The parameter $\varpi$ indicates the channel model. $\varpi = 1$ denotes the case of the non-orthogonal channel access scheme, and $\varpi = 0$ denotes the case of the orthogonal channel access scheme.

{Problem (\ref{eq.dickel}) can be separated into two sub-problems. The sub-problem solved in user $i$ is organized as follows:
\begin{subequations}
\label{eq:sp1}
\begin{align}
\min_{\mathcal{V}_1}\hspace{0.1cm} & -\!\sum_{k\in \mathcal{K}} \check{R}_{i,k} \!+\! \textrm{Tr}\big\{(\mathbf{U}^{n-1}_{1,i})^T\mathbb{P}_{i}^T \mathbb{G}_{i} \big\}  \!+\!\frac{\rho_1}{2}\normf{\mathbb{P}_i^T \mathbb{G}_{i} \!-\!\mathbb{J}_i^{n-1}}^2 \notag\\ 
&\hspace{0.2cm}+\varpi\Big\{\frac{U_{2,i,k}^{n-1}(e_{i,k})}{N} + \frac{\rho_2}{2} (\frac{e_{i,k}-e_{i,k}^{n-1}}{N}+ \theta_{i,k}^{n-1})^2\notag\\ 
&\hspace{0.2cm}+\sum_{j\in \mathcal{I}/\{i\}}\!\![-\frac{U_{2,j,k}^{n-1}\xi_{i,k}}{N} + \frac{\rho_2}{2} (\frac{\theta_{j,k}^{n-1}}{N-1} + \frac{\xi^{n-1}_{i,k}- \xi_{i,k}}{N})^2]\Big\}\\
  \text{s.t.}\hspace{0.1cm}&   \frac{(\norm{{\ddot{\mathbf{Q}}}_{i.k}-\mathbf{q}_{i,k} }^2 + H^2)n_0}{g_0} \leq l_{i,k}, \forall i, k \\
&   (\ref{eq.c1}),(\ref{eq.c2}),(\ref{eq.c4}),(\ref{eq.c10}), (\ref{eq.c18}), (\ref{eq.c19}), \notag\\
& (\ref{eq.c11})\textrm{, if  }\varpi=0, \notag\\
& (\ref{eq.c22}) \textrm{, if  }\varpi=1, \notag
\end{align}
\end{subequations} 
and the  sub-problem solved in the UAV is organized as follows:
\begin{subequations}
\label{eq:sp2}
\begin{align}
\min_{\mathcal{V}_2} \hspace{0.1cm}& \alpha \big[{\sum_{k\in \mathcal{K}}\!\frac{\kappa  V_k^3}{\Delta^2} \! +\!\!\sum_{k\in \mathcal{K}}\!\hat{E}_{k}^{F} }\big] \!-\! \textrm{Tr}\big\{(\mathbf{U}^{n}_1)^T\mathbb{H}\big\} \!+\! \frac{\rho_1}{2}\normf{\mathbb{P}^T \mathbb{G}^{n} \!-\!\mathbb{H}}^2 \\
  \text{s.t.}\hspace{0.1cm}&     \frac{V_k}{\Delta}\leq f^{U}_{max}, \forall k, \\
&(\ref{eq.c7}),(\ref{eq.c9}),(\ref{eq.c14}), (\ref{eq.c20}),(\ref{eq.c21}).\notag
\end{align}
\end{subequations}}
The term $(x)^{n-1}$ represents the variable $x$ obtained in iteration $n-1$. The Lagrangian multipliers $\mathbf{U}_1$ and $\mathbf{U}_2$ are updated at each iteration as follows:
\begin{subequations}
\begin{align}
&\mathbf{U}_1^n = \mathbf{U}_1^{n-1}+\rho_1(\mathbb{P}^T \mathbb{G}^{n} -\mathbb{H})^{n}\label{eq.admm_n_2}\\
&\mathbf{U}_{2,i,k}^n = \mathbf{U}_{2,i,k}^{n-1} +{\rho_2} \theta_{i,k}^n, \label{eq.admm_n_1}
\end{align}
\end{subequations}
where $\theta_{i,k}^n$ is
\begin{align}
\theta_{i,k}^n= \frac{1}{N}(e_{i,k}^n+\xi_{i,k}^n)-\bar{\xi}_{k}^n. \label{eq.admm_n_3}
\end{align}
$\theta_{i,k}$ represents the difference between the user expected interference and the real interference.
At iteration $n$, problem~(\ref{eq:sp1}) is solved by each user individually. The optimization variable set $\mathcal{V}_1$ includes $\{\delta_{i,k}, {W}_{i,k}, \ddot{\mathbf{Q}}_{i,k}, \xi_{i,k},\mathbf{l},\check{\mathbf{R}},e_{i,k}\}$ for all $k\in \mathcal{K}$. To decompose the auxiliary constraint (\ref{eq.admm1}) for each user $i$, we introduce sub-matrices $\mathbb{P}_i$, $\mathbb{H}_i$, and $\mathbf{U}_{1,i}$, which are defined as follows: The parameter matrix $\mathbb{P}_i$ is the sub-matrix sliced from $\mathbb{P}$, where $\mathbb{P}_i = \mathbf{diag} \{1, \chi_i\} $. 
The matrix $\mathbb{J}_i$ is obtained by the information from the UAV, where $\mathbb{J}_i^n = [\mathbf{Q}^n;\mathbf{V}^n/N+\chi_i{W_i}^n-\sum_{j\in \mathcal{I}} \chi_j\textbf{W}^n_{j}/N]$. 
The sub-matrix $\mathbf{U}_{1,i}$ is sliced from the dual variable, where $\mathbf{U}_{1,i} = [\mathbf{U}_{1}(i,:);\mathbf{U}_{1}(N+1,:)]$.
The detailed decomposition process is omitted due to the space limit.
Subsequently, problem (\ref{eq:sp2}) is solved by the UAV. The optimization variable set $\mathcal{V}_2$ includes $\{\mathbf{Q}, \boldsymbol{\omega},\mathbf{A},\hat{\mathbf{E}}^{F}\}$. 
\begin{lemma}
If the initial value of $\{\mathbf{e}^0,\boldsymbol{\xi}^0, \mathbf{U}^0_1,\mathbf{U}^0_2\}$ is shared and unified among all users and the UAV, only information from the UAV required for computing the sub-problem on the user side at each iteration is $\{\mathbb{J}_i^{n-1}, \boldsymbol{\theta}^{n-1}\}$. 
\end{lemma}
\begin{proof}
If the initial value is unified among the UAV and users, the dual variables are not required to be shared and can be computed locally by the UAV and users. For computing the dual variable $\mathbf{U}_{1,i}$ at $n$, the following knowledge is required: 
the updated global value $\mathbb{J}_i^{n-1}$, the historical value for the local information $\mathbb{G}_i^{n-1}$, and the historical value of the dual variable $\mathbf{U}_{1,i}^{n-1}$. 
Therefore, if $\mathbf{U}_{1,i}^{0}$ is identical to all users and the UAV, $\mathbf{U}_{1,i}^{n}$ can be synchronized according to the historical value and the value from the global variable. Similarly,
$\mathbf{U}_{2}$ can be updated by users if the initial value is known.
\end{proof}

\begin{algorithm}
 
\caption{ADMM Algorithm for Solving Problem (\ref{eq.dickel})} 
\label{al.3}
\begin{algorithmic}[1]
\State {Initialize variables $\{\mathbf{e}^0,\boldsymbol{\xi}^0,\boldsymbol{\theta}^0, \mathbb{H}^0, \mathbb{G}^0\}$ and dual variables $\{\mathbf{U}^0_1,\mathbf{U}^0_2\}$. Loop index $n = 0$.}
\Repeat
\State \textbf{{For each user}} $\mathbf{i}$:
\parState 
{\textit{If $\varpi = 0$}: Wait until receive updated $\mathbb{J}_i^{n-1}$. }
\parState{\textit{If $\varpi = 1$}:  Wait until receive updated $\{\mathbb{J}_i^{n-1}, \boldsymbol{\theta}^{n-1}\}$. }
\parState{Calculate the dual variable $\mathbf{U}_{1,i}^{n-1} = \mathbf{U}_{1,i}^{n-2}+\rho_1(\mathbb{P}_i^T \mathbb{G}_i^{n-1} -\mathbb{J}_i^{n-1})$.}
\State Calculate the dual variable $\mathbf{U}_{2}$ for all $i\in \mathcal{I}$ by (\ref{eq.admm_n_1}).
\State Solve problem (\ref{eq:sp1}).
\parState 
{\textit{If $\varpi = 0$}: Send $\mathbb{G}_{i}^{n}$ to the cloudlet.}
\parState {\textit{\textit{If $\varpi = 1$}}: Send $\{\mathbb{G}_{i}^{n},\mathbf{e}_i^{n},\boldsymbol{\xi}_i^{n}\}$ to the cloudlet.}
\State \textbf{{For the UAV-mounted cloudlet}}:
\State Gather information from users to form matrix $\mathbb{G}^{n}$.
\State Solve problem (\ref{eq:sp2}), and update $\mathbb{H}^n$.
\parState {Update dual variable $\mathbf{U}^n_{1}$ by (\ref{eq.admm_n_2}) }
\parState {\textit{If $\varpi=1$}: Update variables $\theta_{i,k}^{n} \forall i,k$ by (\ref{eq.admm_n_3}), and send the variables to users.}
\State $n = n+1$.
\Until{$|\Gamma^n(\mathcal{V},\mathbb{G},\mathbf{V},{\mathbf{U}_1},{\mathbf{U}_2}) -\Gamma^{n-1}(\mathcal{V},\mathbb{G},\mathbf{V},{\mathbf{U}_1},{\mathbf{U}_2})| \leq \theta_3$. }
\end{algorithmic}
\end{algorithm}

Consider the condition in Lemma 4, the distributed algorithm is given in Algorithm \ref{al.3}. In each optimization iteration, user side computes and share matrix $\mathbb{G}$ to the UAV, and UAV computes and shares the matrix $\mathbb{J}$ to users. Meanwhile, when $\varpi = 1$, excepting contributing matrix $\mathbb{G}_i$, user $i$ needs the information $e_{j,k}$ and $\xi_{j,k}$ from other users $j\in \mathcal{I}/\{i\}$ to evaluate the interference. 

By the problem decomposition, at the user side, each user only aims to maximize its own offloading data given the UAV trajectory computed by the UAV-mounted cloudlet and the interference environment in the previous iteration. At the UAV-mounted cloudlet side, the UAV aims to minimize energy consumption under the users' expected UAV trajectories to collect enough workload. The trade-off between the received offloaded tasks and the energy consumption is controlled by the parameter $\alpha$ which is updated out of the ADMM algorithm loop.
Meanwhile, the corresponding variables and constraints are split into two groups. This introduces three main advantages. Firstly, local variables and parameters, such as user location and user offloading constraints, are not required to be uploaded to the UAV. {Similarly, UAV's mechanical parameters and settings are not required to be shared to users for offloading optimization.} Secondly, less configuration is required when the UAV is replaced. Thirdly, the main computation load in solving the problem is from the SOC programming. The SOC constraints are now decomposed and solved by users in parallel such that the computation efficiency can be improved. For ADMM algorithm, in the orthogonal channel model, there are two main distributed blocks: the user side and the UAV side. The convergence of ADMM is guaranteed when the number of blocks is no more than two. In the non-orthogonal channel model, since each user is required to compute the interference variable $e_{i,k}$ parallelly, convergence is not always guaranteed. Proximal Jacobian ADMM can be adopted to ensure the convergence, in which the proximal term  $\frac{\tau}{2} ||x_i-x_i^k||^2$ is further combined in the primal problem of the current algorithm \cite{pjadmm}.

\subsection{Convergence and Complexity Analysis}
{The convergence for the three loops in Algorithms 1 to 3 is guaranteed. For the SCA-based algorithm, if the problem is feasible and the initial values of the approximate variables are in the feasible set of the original optimization problem (\ref{eq:optProb2}), the algorithm convergence is ensured \cite{Boyd}. Moreover, the Dinkelbach algorithm can achieve the optimal $\alpha^*$ with a super-linear rate. }

{The computation complexity of the problem is dominated by the SOC programming \cite{Hu,Wang3}. Suppose that Algorithm~\ref{al.3} runs $L_1\times L_2$ iterations, where the SCA algorithm loop repeats $L_1$ times, and the loop for the Dinkelbach algorithm repeats $L_2$ times. The problem before decomposition, \textit{i.e.}, problem (\ref{eq.dickel}), has $KN$ SOC constraints in 4 dimensions, $K$ SOC constraints in 7 dimensions, and $KN$ SOC constraints in 2 dimensions, where $6KN+4K$ variables participates in those constraints. The overall complexity can be $L_1L_2O\big(\sqrt{2KN+K}(6KN+4K)(20KN+49K+(6KN+4K)^2)\big)$. After ADMM decomposition, for the sub-problem on the user side, there are $K$ SOC constraints in 4 dimensions and $K$ SOC constraints in 2 dimensions. Thus, the computation complexity is $L_1L_2O(1/\theta_3)O\big(\sqrt{2K}(5K)(20K+(5K)^2)\big)$ for each user. On the UAV side, the sub-problem contains $K$ SOC constraints in 7 dimensions. The complexity is $L_1L_2O(1/\theta_3)O\big(\sqrt{K}(2K)(49K+(2K)^2)\big)$.}


\section{Proactive Trajectory Design Based on Spatial Distribution Estimation}
\label{sec6}
So far, we have introduced the trajectory design and resource allocation for the scenario that all computation load information and user location are known. {However, some IoT nodes have a certain mobility \cite{Hakiri}.} It is hard for users to know their future positions during the upcoming computation cycle. Moreover, users needs to send the offloading requests at the beginning of the cycle. It means that the user may buffer the computation task until a new cycle begins, which introduces extra delay for waiting to send the request. {Thus, the maximum queue delay may reach to $T$ seconds.}
To deal with the above issues, in this subsection, we introduce an approach to estimate the spatial distribution of user locations in a cycle. The mobility of users is predicted by an unsupervised learning tool, kernel density estimation method \cite{KDE}, and the computation load of each user is considered in a stochastic model correspondingly. The UAV trajectory is optimized via the estimated knowledge about ground users. {Thus, UAV can collect the offloaded tasks of users without requesting in advance.}

To estimate the location of users, each user need to report its current location periodically. The sampled location of user $i$ is represented by $q_i$. We use the sampled location to estimate the spatial distribution of users for the cycle, where the probability density function for the user at $(x,y)$ is denoted as $f(x,y)$. 

In order to compute $f(x,y)$, consider a small region $\textit{R}$ which is a rectangle area with side length of $h_x$ and $h_y$, \textit{i.e.}, Parzen window. 
To count the number of users falling within the region, we define the following function to indicate if user $i$ is in the area:
\begin{equation}
C(q^x_i,q^y_i;\textit{R}) = \left\{
                \begin{array}{ll}
                  1, \textrm{if } \max\{\frac{||q^x_i-x||}{h_x},\frac{||q^y_i-y||}{h_y}\}\leq \frac{1}{2}\\
                  0, \textrm{otherwise,}
                \end{array}
              \right.
\end{equation}
where $(x,y)$ is the central point of the area. Thus, for a large $N$, the general expression for non-parametric density estimation is \cite{KDE}
 \begin{equation}
 f(x,y) = \frac{1}{Nh_xh_y} \sum_{i\in\mathcal{I}} C(q^x_i,q^y_i;\textit{R}).
 \end{equation}
To establish continuous estimation function, a smooth Gaussian kernel is applied, where
 \begin{equation}
 \hat{f}(x,y) = \frac{1}{N\sqrt{h_xh_y}} \sum_{i\in\mathcal{I}} \frac{1}{2\pi} e^{-[\frac{(q^x_i-x)^2}{2h_x}+\frac{(q^y_i-y)^2}{2h_y}]}. \label{eq.ker}
 \end{equation}
The term $\hat{f}(x,y)$ is the distribution of Gaussian kernel estimation. In (\ref{eq.ker}), $h_x$ and $h_y$ represent the bandwidth of the Gaussian kernel rather than the side length of the Parzen window.
To improve the estimation quality, the proper bandwidth, $h_x$ and $h_y$, needs to be selected to minimize the error between the estimated density and the true density. In this work, the maximum likelihood cross-validation method \cite{KDE,Mozaffari} is adopted to determine the bandwidth $h_x$ and $h_y$. The optimal bandwidth is 
\begin{equation}
[h_x^*,h_y^*] = \textrm{argmax}\{\frac{1}{N} \sum_{i\in \mathcal{I}} \log \hat{f}_{-i}(q_i^x,q_i^y)\},
\end{equation} 
where $\hat{f}_{-i}(q_i^x,q_i^y)$ is the estimated distribution {in which user $i$ is left out of the estimation}. 
In order to apply the estimated distribution into our proposed approach, 
we divide the working area of the UAV $\mathcal{A}$ into $G\times G$ sub-areas. For each sub-area $\mathcal{A}_i$, there is a virtual user located at the center of the area. The virtual user carries all the computation tasks in the sub-area. It is assumed that the distribution of the task input data size and user spatial location are independent. The expected length of input bits for the tasks generated by a user by $\mathbb{E}[X]$. Thus, the expected length of computing bits generated inside the sub-area $\mathcal{A}_i$ is
\begin{equation}
\mathbb{E}[I_i] = \mathbb{E}[X]\mathbb{E}[N_i] = \mathbb{E}[X] \int_{(x,y)\in \mathcal{A}_i} \hat{f}(x,y) dxdy,
\end{equation}
where $\mathbb{E}[N_i] $ denotes the expected number of users in the sub-area $\mathcal{A}_i$.  Our proposed approach can now be adopted to solve the problem: In the new problem, there are $G^2$ virtual users participating in the computation task offloading, and virtual user $i$ has $\mathbb{E}[I_i]$ computation load to be done in a cycle.  The location of user $i$ is fixed at the center of the sub-area. For the orthogonal channel model, the virtual user $i$ shares a portion of $\mathbb{E}[N_i]/N$ of the channel bandwidth. As $G$ increases, the performance of the estimation will be improved correspondingly.

\section{Numerical Results}
\label{sec7}
In this section, {we evaluate the performance of our proposed optimization approach.} The parameter settings are given in Table \ref{t1}. 
The channel gain parameter $g_0$ is -70 dB. Let the term $p$ represent the percentage of computation tasks that have to be offloaded to the cloudlet, \textit{i.e.}, $p$ = $(\check{I}_i/I_i)*100\%$. We consider that users have homogeneous offloading requirements in the simulation, \textit{i.e.}, $E_i^T$ and $p$ are identical for all user. The term ``NO" represents the non-orthogonal channel access scheme, and the term ``O" represents the orthogonal channel access scheme. We also consider the circular trajectory scheme as the benchmark, where the UAV moves around a circle within a cycle, with the circle center located at (0.5,0.5) km, and the radius is predefined.
{Two network scenarios are considered: a three-node scenario and a four-node scenario. In the three-node scenario, there are three users located at (0,1) km, (1,1) km, and (1,0) km, as shown in Fig. \ref{fig.nr1}(a). At the beginning of the cycle, the UAV moves from the location (0,0) at an initial speed  (-10,0) m/s. By the end of the cycle, the UAV returns to the final designated position at (0.5,0) km. 
In the four-node scenario, there are four users located at the randomly generated locations. The users travel at constant speeds which are random selected from [-3,-3] m/s to [3,3] m/s, as shown in Fig. \ref{fig.nr1}(b). The UAV moves from the location (200,200) m at an initial speed (-10,0) m/s and returns to the initial position at the end of the cycle.
}

\begin{table}[t]  
\centering
\caption{Parameter Settings for the Three-node Scenario}
\begin{tabular}{cc|cc}
\hline
Parameter & Value       & Parameter  & Value       \\ \hline
$B$       & 3 MHz      & $\kappa$   & $10^{-28}$  \\
$\sigma^2$  &-80 dBm/Hz &  $\gamma_1$ & 0.0037      \\
$\chi_i$  & 1550.7      & $\gamma_2$ & 500.206      \\
$\Delta$  & 1.5 s       & $H$        & 100 m         \\
$a_{\textrm{max}}$ & 50 m/s$^2$      & $P$        & 100 mW          \\
$v_{max}$ & 35 m/s      & $K$        & 50 \\ \hline
\end{tabular}
\label{t1}
\end{table} 
\begin{figure}[t]  
 \centering  
 \hspace*{-0.5cm} 
  \includegraphics[width=96mm]{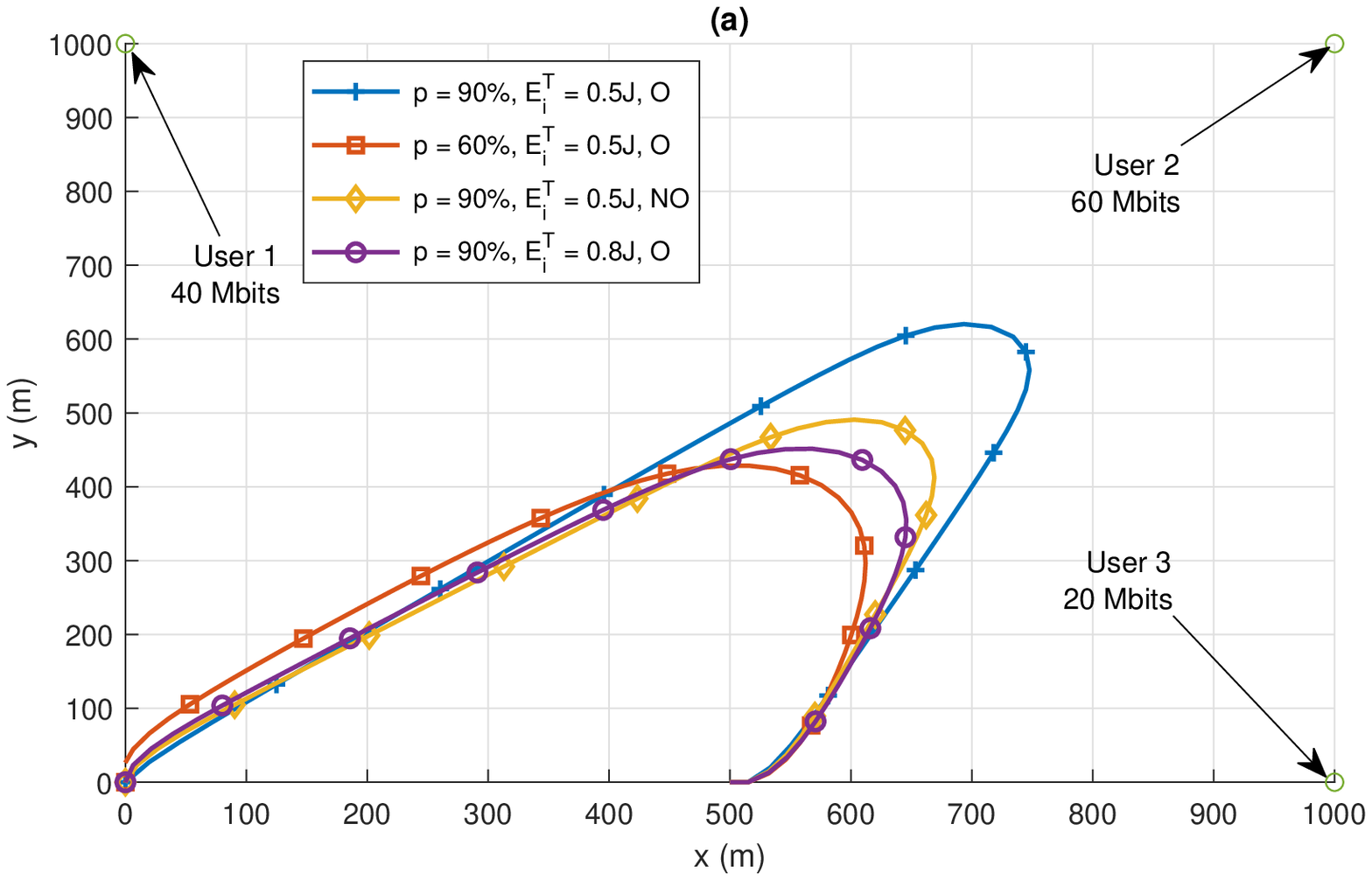}\\  
   \hspace*{-0.5cm}
   \includegraphics[width=96mm]{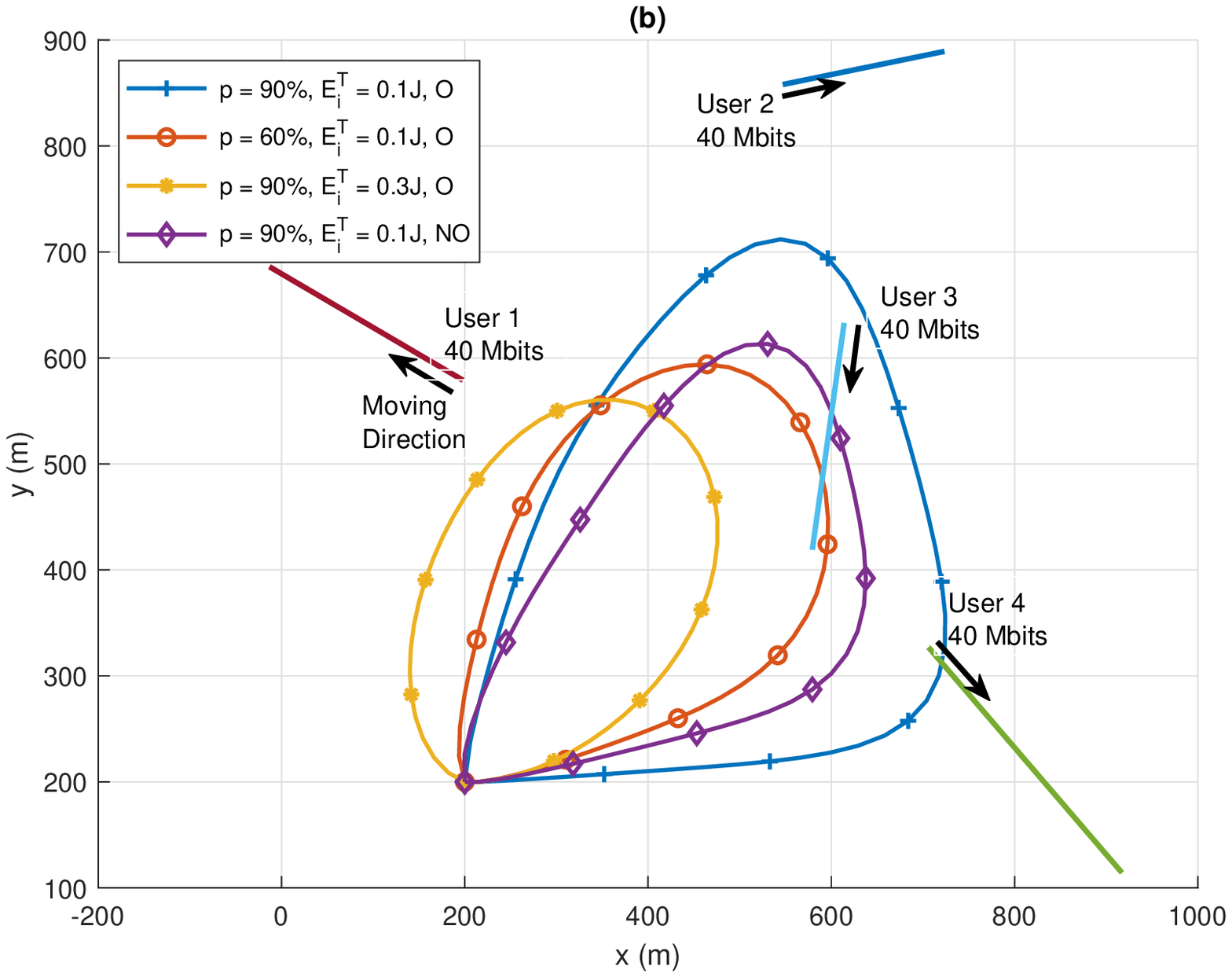}\\  
    \vspace*{-0.1cm}
  \caption{{Optimal UAV trajectories with different parameter settings: (a) the three-node scenario; (b) the four-node scenario with user mobility, where the solid straight lines represent user trajectories, and the arrows represent user moving directions.}}   \label{fig.nr1}
\end{figure}

{The UAV trajectory results obtained by the proposed approach are shown in Fig. \ref{fig.nr1}.} In the three-node case shown in Fig. \ref{fig.nr1}(a), the UAV takes most of the time moving towards and stays around the location of user 2 due to high computation task loads of the user. With a higher minimum offloading requirement $p$, the UAV moves closer to users in order to collect more offloading tasks. Similarly, with a lower maximum communication energy requirement $E_i^T$, the UAV also moves closer to users to reduce the user's offloading communication energy consumption. Moreover, since the non-orthogonal access method has a higher channel capacity, under the same condition, the trajectory of the non-orthogonal case is shrunk to preserve the mechanical energy consumption compared to the orthogonal channel case. {Similar results can be obtained in the four-node case, as shown in Fig. \ref{fig.nr1}(b).}

\begin{figure}[t]  
 \centering  
  \includegraphics[width=78mm]{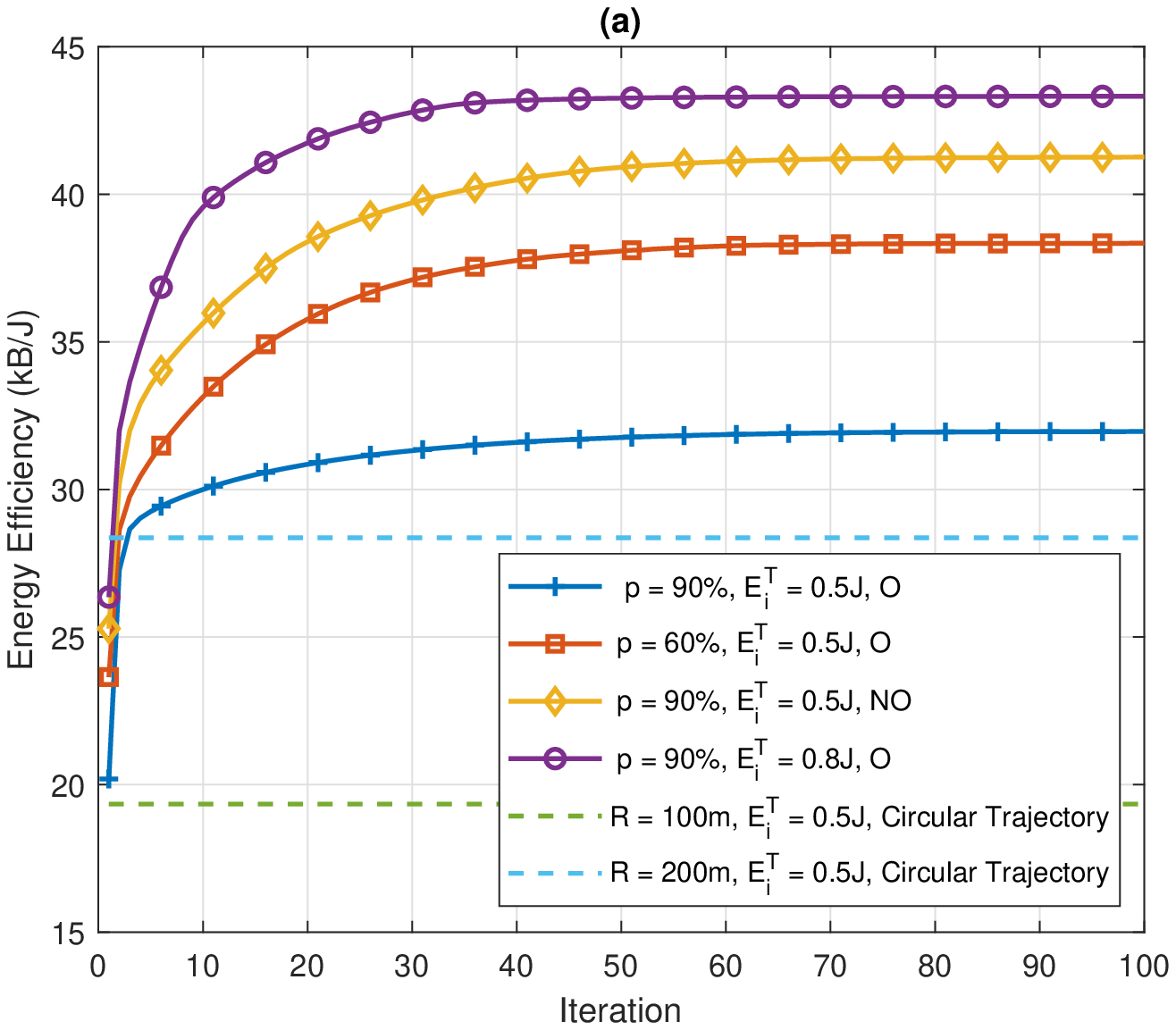}\\  
   \includegraphics[width=78mm]{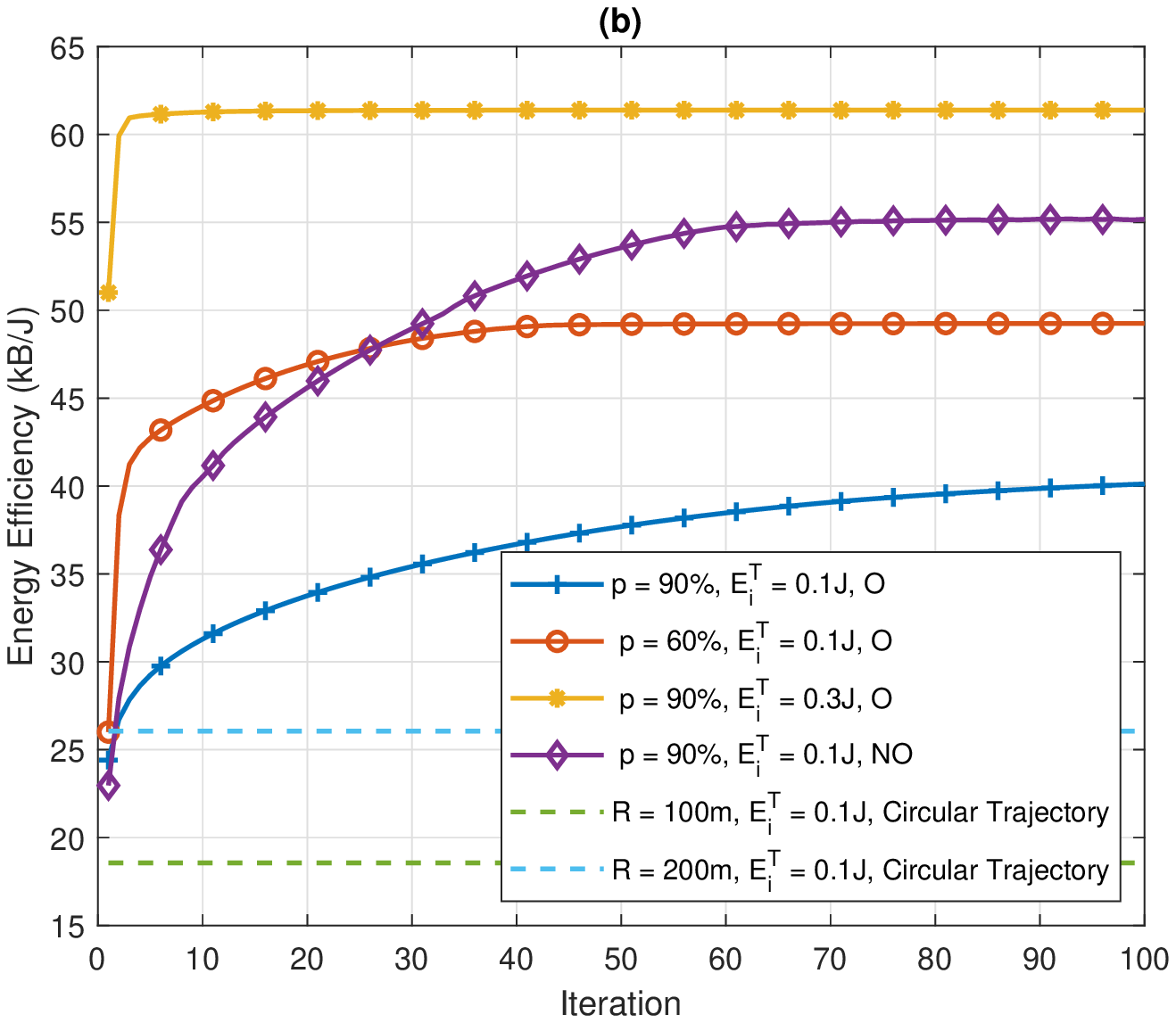}\\  
    \vspace*{-0.1cm}
  \caption{{Energy efficiency versus main loop iteration number with different trajectory designs: (a) the three-node scenario; (b) the four-node scenario with user mobility.} } \label{fig.nr3}
\end{figure}

{The comparisons of the energy efficiency with different settings are shown in Fig. \ref{fig.nr3}. In Figs. \ref{fig.nr3}(a) and \ref{fig.nr3}(b), the x-axis represents the iteration number of the SCA-based algorithm loop. As shown in Fig. \ref{fig.nr3}(a), the energy efficiency converges at $t = 30$ in the three-node scenario, while the number of iterations till convergence is increased in the four-node scenario. Moreover, for both scenarios, with loose user offloading requirements, the energy efficiency is improved due to the expanded optimization feasible set. In contrast, with tight user offloading requirements, the energy efficiency is decreased significantly due to high energy consumption for the UAV to move closer to the users. }

\begin{table}[t]  
\caption{Parameter Setting for Fig. \ref{fig.nr6}}
\centering
\label{t2}
\scriptsize{
\begin{tabular}{ccc|ccc|ccc}
\hline
Index & $p$  & $E_i^T$ & Index & $p$  & $E_i^T$ & Index & Radius & $E_i^T$ \\ \hline
1     & 90\% & 0.5 J   & 4     & 60\% & 0.5 J   & 7     & 200 m  & 0.5 J   \\ \hline
2     & 90\% & 0.8 J   & 5     & 60\% & 0.8 J   & 8     & 200 m  & 0.8 J   \\ \hline
3     & 90\% & 1.1 J   & 6     & 60\% & 1.1 J   & 9     & 200 m  & 1.1 J   \\ \hline
\end{tabular}}
\end{table}

\begin{figure}[t]  
 \centering   
  \includegraphics[width=78mm]{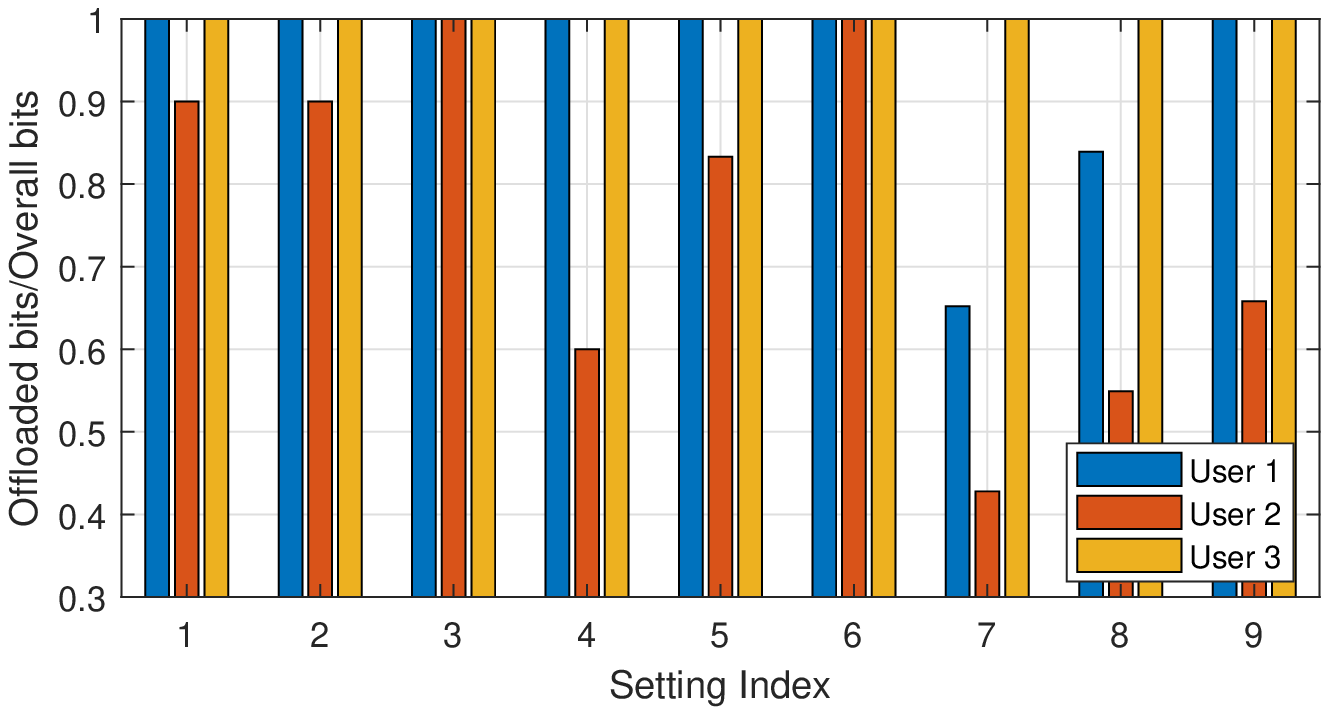}\\  
    \vspace*{-0.1cm}
  \caption{The ratio between the offloaded task data amount and the overall computation task data among generated by users with different parameter settings. } \label{fig.nr6} 
\end{figure} 

For the three-node case, the ratio between the offloaded data amount and the overall computing data amount is shown in Fig. \ref{fig.nr6}. The parameter setting for the indexes are given in Table \ref{t2}, where the results by the proposed approach are shown in 1-6, and the results by the circular trajectory are shown in 7-9. For all scenarios, the proposed approach can achieve the minimum offloading requirement, while the circular trajectory scheme cannot guarantee to achieve the requirement. Moreover, when the maximum communication energy requirement $E_i^T$ is increased, the UAV can collect more data even though its trajectory is far away from users compared to the case with a low $E_i^T$. 
The UAV also collects the extra offloaded tasks, which is beyond the users' requirement, to improve its energy efficiency. 

\begin{figure}[t]  
 \centering   
   \vspace*{-0.3cm}
  \includegraphics[width=78mm]{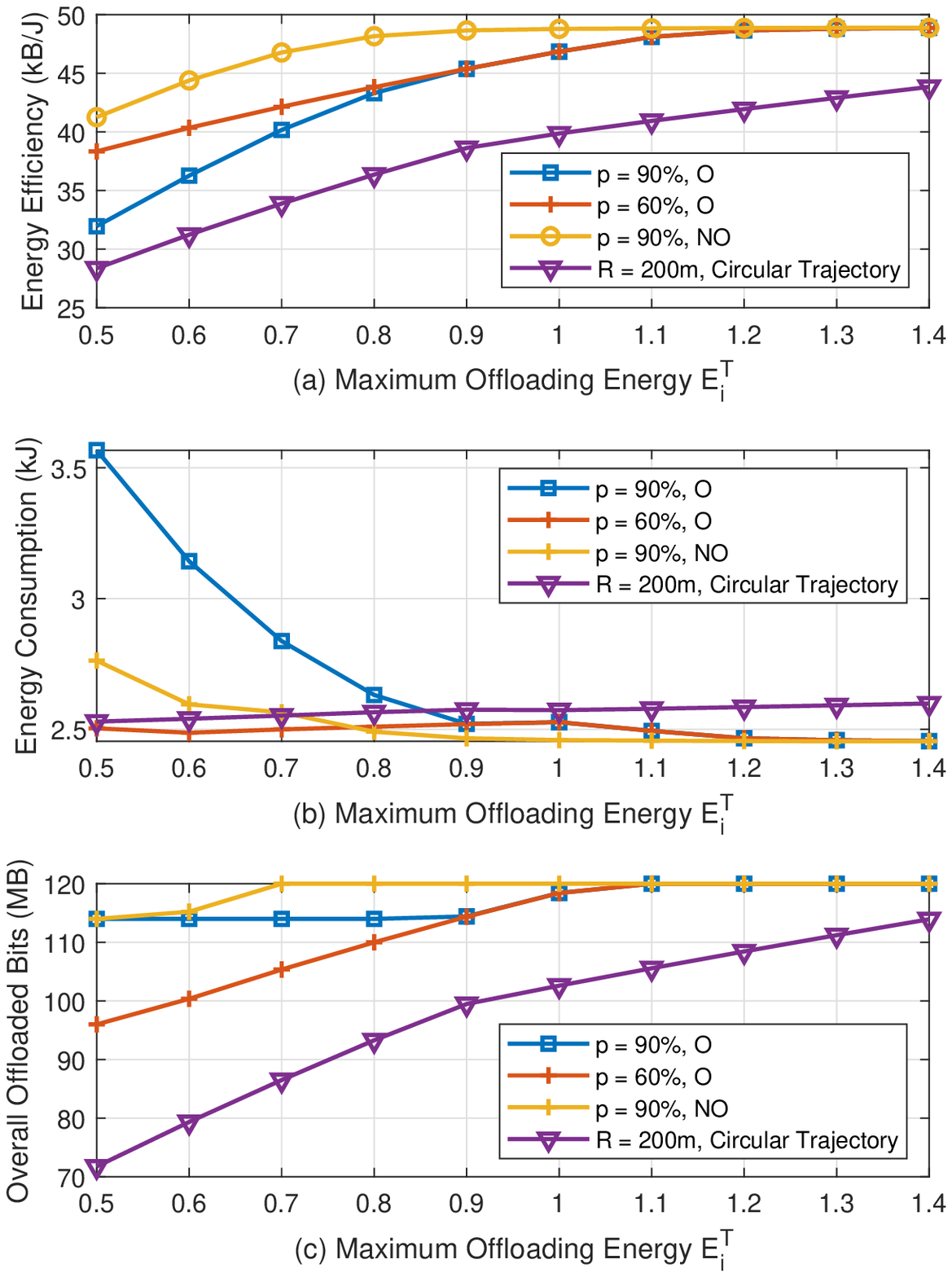}\\  
    \vspace*{-0.3cm}
  \caption{(a) Energy efficiency versus the maximum offloading communication energy with different settings. (b) Overall energy consumption in a cycle versus the maximum offloading communication energy. (c) Overall offloaded bits in a cycle versus the maximum offloading communication energy.} 
  \label{fig.nr9} 
\end{figure} 

The trade-off between the maximum offloading energy, \textit{i.e.}, $E_i^T$, and the energy efficiency in the three-node case is shown in Fig. \ref{fig.nr9}(a). As $E_i^T$ increases, the energy efficiency of the UAV is increased at first and hits the ceiling in a high $E_i^T$. 
At that point, $E_i^T$ is not the factor that limits the energy efficiency performance since all user's computing data is collected as shown in Fig. \ref{fig.nr9}(c). 
When the energy efficiency reaches the maximum value, the UAV will find a path that has minimum energy consumption given that all tasks are offloaded. Furthermore, our proposed approach can improve the energy efficiency significantly compared to the circular trajectory.

The magnitudes of the UAV acceleration and velocity in the three-node case are shown in Fig. \ref{fig.nr2}(a) and Fig. \ref{fig.nr2}(b), respectively. The final velocity is constrained to be equal to the initial velocity.
Note that the optimal velocity cannot be zero due to the characteristic of fixed-wing UAV.
With the lower maximum energy requirement, both magnitudes of acceleration and velocity are increased, such that the UAV can move closer to users. 
With the higher energy requirement, the fluctuation on velocity and acceleration decreases to reduce the propulsion energy consumption of the UAV.

\begin{figure}[t]  
 \centering   
   \vspace*{-0.3cm}
  \includegraphics[width=78mm]{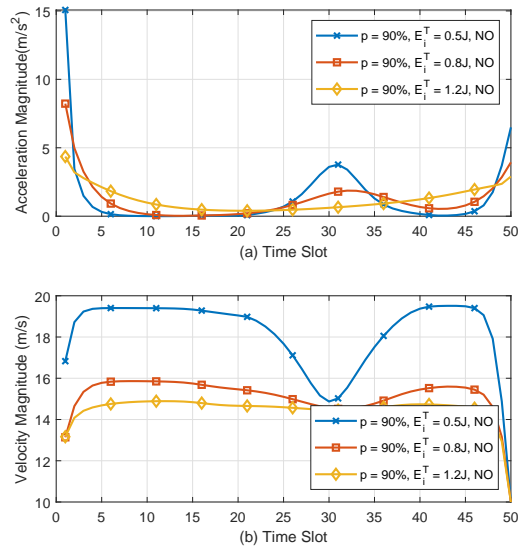}\\  
    \vspace*{-0.3cm}
  \caption{(a) The acceleration of the UAV in the cycle. (b) The speed of the UAV in the cycle.}  \label{fig.nr2}
\end{figure} 

The ratio of the actual allocated transmit power to the maximum power, $\delta_{i,k}$, for the three users in a cycle is shown in Fig.~ \ref{fig.nr4}(a). Note that the overall offloading communication energy is limited. For the user with high offloading demands, \textit{i.e.}, user 2, the ratio is maximized when the UAV moves adjacent to it, while the ratio is minimized when the UAV moves away from it. The user tends to preserve the communication energy and starts the offloading only when the data rate is high. However, for user 3, the transmit power is still allocated when the UAV is far away from the location of the user for two reasons: Firstly, the maximum communication energy of the user allows user uploading the data even though the user transmission efficiency is low. Secondly, the UAV-mounted cloudlet prefers collecting the data in advance such that it can balance the computation load to reduce the computing energy cost. 
The computation load allocation of the cloudlet in the three-node case is shown in Fig. \ref{fig.nr4}(b). Since the energy consumption is cubically increased as the computation load in a unit time increased (based on (\ref{eq.c6}) and (\ref{eq.c60})), the computation load is preferred to be balanced among time slots. However, the computation load can only be executed after the corresponding tasks are offloaded into the cloudlet. Therefore, in the case with limited maximum communication energy, the allocated computation load is increased only when the new offloaded tasks are received. In contrast, with the loose maximum communication energy constraint, the workload fluctuation is reduced significantly to minimize the computing energy consumption. 

\begin{figure}[t]  
 \centering   
   \vspace*{-0.5cm}
  \includegraphics[width=78mm]{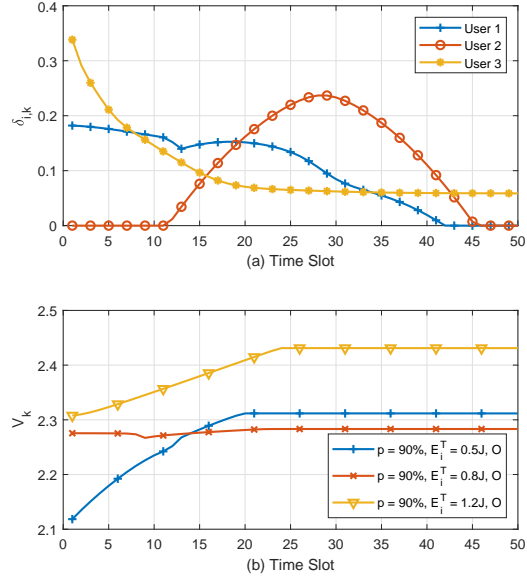}\\  
    \vspace*{-0.3cm}
  \caption{(a) The transmit power allocation among three users, where $p$ = 90\%, and $E_i^T$ = 0.5 J under orthogonal channel scenario. (b) The workload allocation with different settings.}   \label{fig.nr4}
\end{figure}

\section{Conclusions}
\label{sec8}
In this paper, an optimization approach has been proposed to maximize the energy efficiency of a UAV-assisted MEC system, where the UAV trajectory design and resource allocation have been jointly considered. The non-convex and non-linear energy efficiency maximization problem has been solved in a distributed manner. Moreover, the node mobility estimation has been adopted to design a proactive UAV trajectory when the knowledge of user trajectory is limited. 
Our work can offer valuable insights on UAV optimal trajectory design for providing on-demand edge computing service for remote IoT nodes. In the future, considering the uncertainty of user mobility and the time-invariant computation demand, we will focus on the online resource management in UAV-assisted MEC system under a dynamic channel environment.

 \section*{Appendix A: Proof of Lemma~\ref{lemma.1}}
Firstly, to deal with the non-convex function on the numerator,  \textit{i.e.}, $R_{i,k}(\delta_{i,k},\mathbf{Q}_k)$, we introduce the auxiliary variable $\check{R}_{i,k}$ to indicate the lower bound of the data rate for user $i$ in slot $k$.   Moreover, we introduce two auxiliary variables: the term $\xi_{i,k}$, where $\xi_{i,k} \leq \delta_{i,k}P/l_{i,k}$, and the term $l_{i,k}$, where $l_{i,k} \geq N_0/h_{i,k}$. Thus, the following relation can be established
\begin{equation}
\check{R}_{i,k} \leq \frac{B \Delta}{N} \log(1+\xi_{i,k}) \leq R_{i,k}(\delta_{i,k},\mathbf{Q}_k),
\end{equation}
where $\check{R}_{i,k}$ is the epigraph form of $R_{i,k}(\delta_{i,k},\mathbf{Q}_k)$. 
When (\ref{obj:optProb3}) is maximized, \textit{i.e.}, the numerator $\check{R}^*_{i,k}$ is maximized, we have $l^*_{i,k} = 1/g^*_{i,k}$, $\xi^*_{i,k} = \delta^*_{i,k}P/l^*_{i,k}$, and $\check{R}^*_{i,k} = {R}_{i,k}(\delta^*_{i,k},\mathbf{Q}_k^*)$.

Furthermore, to deal with the non-linear function on the denominator, \textit{i.e.}, $E_{k}^{F}(\mathbf{Q})$, we introduce an auxiliary variable $\hat{E}^F_{k}$ to indicate the upper bound of the UAV propulsion
energy in slot $k$. For the non-linear part of the function, we introduce two auxiliary variables: the term $\omega_k$, where $\omega_k^2\leq {\norm{{\mathbf{v}_k}(\mathbf{Q})}}^2$, and the term $A_{i,k}$, where $A_{i,k}\geq (1/\omega_k)(1+{\norm{a_k(\mathbf{Q})}^2}/{g^2})$. Thus, we have
\begin{align}
\hat{E}^F_{k}& \geq \gamma_1 \norm{{\mathbf{v}_k}(\mathbf{Q})}^3 + \gamma_2 A_k \notag \\
&\geq \gamma_1\norm{{\mathbf{v}_k}(\mathbf{Q})}^3 + \gamma_2  \frac{1}{\omega_k}(1+\frac{\norm{a_k(\mathbf{Q})}^2}{g^2})\geq E_{k}^{F}(\mathbf{Q}).
\end{align}
Similarly, when (\ref{obj:optProb3}) is maximized, \textit{i.e.}, the denominator $E_{k}^{F}(\mathbf{Q})$ is minimized,  $\hat{E}^{F*}_{k} = E_{k}^{F}(\mathbf{Q}^*)$. Therefore, problem (\ref{eq:optProb3}) is equivalent to problem (\ref{eq:optProb2}), and $\eta^* = \check{\eta}^*$
 \section*{Appendix B: Proof of Lemma~\ref{lemma.2}}
Constraint (\ref{eq.c12}) can be transformed into the following equivalent form:
\begin{equation}
(\xi_{i,k}+l_{i,k})^2-(\xi_{i,k}-l_{i,k})^2 \leq 4\delta_{i,k}P,
\end{equation}
which is difference of convex functions \cite{Boyd}. Then, we approximate the second part of the equation by the Taylor expansion:
\begin{align}
& (\xi_{i,k}-l_{i,k})^2  \approx (\xi^t_{i,k}-l^t_{i,k})^2 + \begin{bmatrix}
        2\xi^t_{i,k}-2l^t_{i,k} \\
        2l^t_{i,k}-2\xi^t_{i,k}
    \end{bmatrix}^T \begin{bmatrix}
        \xi_{i,k}-\xi^t_{i,k} \\
        l_{i,k}-l^t_{i,k}
    \end{bmatrix} 
\end{align}
Then, we further reformulate the approximated equation as the constraints shown in (\ref{eq.c19}) with a cone expression. Moreover, constraint (\ref{eq.c16}) is approximated by constraint (\ref{eq.c20}) in a similar way. Constraints (\ref{eq.c15}) and (\ref{eq.c11_no}) are approximated by (\ref{eq.c21}) and (\ref{eq.c22}) respectively by first order Taylor expansion to obtain the lower bound on the squared norm and the subtracted term, respectively.

All the approximated constraints (\ref{eq.c19})-(\ref{eq.c22}) are stricter than their original counterparts, guaranteeing that the solution of the approximated problem is strictly smaller than the original optimum. For example, consider the optimal $\xi_{i,k}$ and $l_{i,k}$ obtained by solving the approximated problem, which is denoted by $\xi^a_{i,k}$ and $l^a_{i,k}$. 
These two variables are bounded by constraint (\ref{eq.c19}) in the approximated problem. Comparing (\ref{eq.c19}) with the original constraint (\ref{eq.c12}) and considering the property of the Taylor expansion, we have $\xi^a_{i,k}l^a_{i,k} + \Delta_{approx} \leq \delta_{i,k}P$, where $\Delta_{approx} \geq 0$. Thus,
\begin{equation}
\frac{B \Delta}{N} \log(1+\xi_{i,k}^a) \leq \frac{B \Delta}{N} \log(1+\frac{\delta_{i,k}P}{l^a_{i,k}})
\end{equation}
Moreover, due to $l^a_{i,k}\geq 1/g_{i,k}$, we have
\begin{equation}
\frac{B \Delta}{N} \log(1+\frac{\delta_{i,k}P}{l^a_{i,k}})\leq R_{i,k}(\delta_{i,k}, \mathbf{Q}_k) .
\end{equation}
Therefore, the approximation on constraint (\ref{eq.c12}) will leads to $\check{R}_{i,k}^* <R_{i,k}(\delta_{i,k}, \mathbf{Q}_k) $. Other approximated constraints can be proven similarly to show that the proposed approximated objective function provides the global lower bound for original objective function (\ref{eq:optProb2}). Moreover,  due to the gradient consistency in the first order estimation, the SCA algorithm will be stopped when a local optimizer is found.

{\footnotesize
\bibliographystyle{IEEEbib}
\bibliography{reference}
}

\end{document}